%% file: main.tex
\input{config/preamble}
\input{config/commands}

\title{WKB Spectral Asymptotics for a One-Dimensional Dirac Operator with a Slowly Varying Mass Profile}
\author{Owen Sutton\thanks{University of Minnesota School of Mathematics, sutto331@umn.edu}
  \; and Alexander B. Watson\thanks{University of Minnesota School of Mathematics, abwatson@umn.edu}}

\begin{document}

\maketitle


   









\begin{abstract}
  We study the semiclassical spectral theory of a one-dimensional Dirac operator $D_\varepsilon = m(x) \sigma_x - i \varepsilon \partial_x \sigma_y$, where $m(x)$ is a slowly varying mass profile with a sign change encoding a topological phase transition and $0 < \varepsilon \ll 1$ is the semiclassical parameter. We derive a modified Bohr-Sommerfeld quantization condition for the squared operator $D_\varepsilon^2$ via a systematic formal WKB construction producing approximate solutions with $\mathcal{O}(\varepsilon)$ error. This condition takes the form  
\begin{equation*}
    \frac{1}{\varepsilon}
        \int_{x_-}^{x_+}
        \sqrt{\mathcal{E} - m(x)^2} \, \diff x
    = \left( n + \frac{\sigma_\mathcal{D} + 1}{2} \right) \pi, 
        \quad n = 0, 1, 2, \ldots,
\end{equation*}
where $\mathcal{E}$ is the eigenvalue and $x_\pm$ are the classical turning points. Our result differs from the standard result by the half-integer shift depending on the pseudo-spin index $\sigma_\mathcal{D} \in \{\pm 1\}$. Setting $\sigma_\mathcal{D} = -1$ and $n = 0$ yields $\mathcal{E} = 0$, recovering the topologically protected zero mode. The condition is asymptotically valid to leading order in $\varepsilon$, with relative errors of $O(\varepsilon)$, and is verified to be exact for the P\"{o}schl--Teller potential. We provide numerical computations confirming the convergence of eigenvalues and eigenfunctions to their WKB approximations.
\end{abstract}

\tableofcontents

\section{Introduction}
\label{sec:S2_setup}

Propagating edge states at the boundary between two topologically distinct two-dimensional materials can be modeled by an effective one-dimensional Dirac operator; see, e.g., \cite{2013FruchartCarpentier}. Edge states have attracted huge attention both in the mathematical literature (see, e.g., \cite{Zhu2024,BalBecker,Drouot2021,2004KellendonkSchulz-Baldes_2,2018GrafShapiro,LuWatsonWeinstein,fefferman_leethorp_weinstein_memoirs,Bal_CMS,Graf2021,2017DelplaceMarstonVenaille,Ludewig2020}) and in applications (see, e.g., \cite{2019Ozawaetal,Gilbert2021} and references therein) in recent years because of their remarkable robustness to perturbations.

In this work we consider the semiclassical limit of this Dirac operator, introducing a small parameter $0 < \varepsilon \ll 1$ and writing the operator in the form
\begin{equation} \label{eq:S2_dirac}
    D_\varepsilon
    := m(x) \sigma_x
        - i \varepsilon \partial_x \sigma_y
    = \begin{bmatrix}
        0 & m(x) - \varepsilon \partial_x \\
        m(x) + \varepsilon \partial_x & 0
    \end{bmatrix}.
\end{equation}
The parameter $\varepsilon$ controls the semiclassical regime. Equivalently, this scaling corresponds to assuming that the mass profile $m(x)$ varies slowly on the scale of the wave oscillations, has characteristic magnitude independent of $\varepsilon$, and approaches constants at spatial infinity. We assume $m \in C^1(\mathbb{R})$, bounded with bounded derivative and satisfies
\begin{equation} \label{eq:S2_mass_req}
    \lim_{x \to \pm \infty} m(x)
    = \pm 1, \qquad
    m(0) = 0.
\end{equation}
The sign change of $m(x)$ across the interface encodes the change in the underlying topological invariant.

We consider the spectral problem
\begin{equation} \label{eq:1}
    D_\varepsilon \Psi = E \Psi,
  \end{equation}
and are interested in constructing approximate eigenfunctions and eigenvalues when $E$ lies in the gap between essential spectrum $(-1,1)$ using WKB theory. For this construction to go through, we assume further that for each such energy $E$, the equation $m(x)^2 = E$ has exactly two simple solutions $x_- < x_+$ and that these solutions satisfy
\begin{equation}
    \partial_x (m(x)^2 - E)\big|_{x = x_\pm}
    = 2 m(x_\pm) m'(x_\pm) \neq 0.
\end{equation}
This nondegeneracy condition is the precise hypothesis required for the Airy-type (turning point) WKB analysis. Equivalently, we assume that $m(x)$ is strictly monotone increasing.
  
To facilitate the use of WKB methods, we consider the squared operator $D_\varepsilon^2$, which relates to the original spectrum via the following symmetry properties.
\begin{proposition}[Spectral Symmetry]
\label{prop:S2_symmetry}
The operator $D_\varepsilon$ is self-adjoint with a real spectrum. The spectrum is symmetric about the origin: if $E \neq 0$ is an eigenvalue of $D_\varepsilon$, then $-E$ is also an eigenvalue. Consequently, non-zero eigenvalues of $D_\varepsilon^2$ are at least two-fold degenerate. Furthermore, if $D_\varepsilon^2$ possesses a simple eigenvalue at $E=0$, the associated eigenfunction is a simultaneous eigenfunction of the unitary operator $U := -i\sigma_y$.
\end{proposition}

\begin{proof}
First, note that $U$ is unitary since $UU^* = (-i\sigma_y) (i\sigma_y^*) = \sigma_y \sigma_y = I$, using $\sigma_y^* = \sigma_y$ and $\sigma_y^2 = I$. The operator $D_\varepsilon$ anti-commutes with $U = -i\sigma_y$. We first verify the anti-commutation relation $\{D_\varepsilon, U\} = 0$.
With $U = -i\sigma_y = \begin{bsmallmatrix} 0 & -1 \\ 1 & 0 \end{bsmallmatrix}$,
direct computation gives
\begin{equation*}
    U D_\varepsilon
    = \begin{bmatrix} 0 & -1 \\ 1 & 0 \end{bmatrix}
    \begin{bmatrix} 0 & m - \varepsilon \partial_x \\
        m + \varepsilon \partial_x & 0 \end{bmatrix}
    = \begin{bmatrix}
        -(m + \varepsilon\partial_x) & 0 \\
        0 & m - \varepsilon\partial_x
    \end{bmatrix},
\end{equation*}
\begin{equation*}
    D_\varepsilon U
    = \begin{bmatrix} 0 & m - \varepsilon \partial_x \\
        m + \varepsilon \partial_x & 0 \end{bmatrix}
    \begin{bmatrix} 0 & -1 \\ 1 & 0 \end{bmatrix}
    = \begin{bmatrix}
        m - \varepsilon \partial_x & 0 \\
        0 & -(m + \varepsilon \partial_x)
    \end{bmatrix},
\end{equation*}
so $U D_\varepsilon + D_\varepsilon U = 0$, confirming $\{D_\varepsilon, U\} = 0$.

Let $(\mu, f)$ be an eigenvalue-eigenfunction pair of $D_\varepsilon$. Then,
\begin{equation*}
    U D_\varepsilon f
    = U \mu f
    \implies
    D_\varepsilon U f
    = - \mu U f,
\end{equation*}
which shows that $-\mu$ is an eigenvalue with eigenfunction $Uf$. Since $D_\varepsilon$ is self-adjoint, its spectrum lies on the real axis, and the spectrum of $D_\varepsilon^2$ is contained in $[0, \infty)$. It follows that
\begin{equation*}
    E^2 \neq 0 \in \sigma(D_\varepsilon^2)
    \implies \{+E, -E\} \subset \sigma(D_\varepsilon).
\end{equation*}
Thus, non-zero eigenvalues of $D_\varepsilon^2$ are two-fold degenerate. For a simple eigenvalue at $E=0$, the anti-commutation relation implies the eigenfunction must be invariant (up to a phase) under $U$, making it a simultaneous eigenfunction.
\end{proof}

By squaring the Dirac operator, the first-order system decouples into two scalar second-order equations.

\begin{proposition}[Squared Dirac Operator]
\label{prop:S2_squared}
The squared Dirac operator $D_\varepsilon^2$ is diagonal and takes the form
\begin{equation} \label{eq:S2_dirac_squared}
    D_\varepsilon^2
    = \operatorname{diag} \left(
        - \varepsilon^2 \partial_x^2
            + m(x)^2 - \varepsilon m'(x), \;
        - \varepsilon^2 \partial_x^2
            + m(x)^2 + \varepsilon m'(x)
    \right).
\end{equation}
Each diagonal component is a one-dimensional semiclassical Schr\"{o}dinger operator with principal symbol $p(x, \xi) = \xi^2 + m(x)^2$ (the leading-order term in the symbol expansion, corresponding to the classical Hamiltonian $H = \xi^2 + m(x)^2$ on phase space) and a sub-leading correction $\pm \varepsilon m'(x)$.
\end{proposition}

\begin{proof}
Direct matrix multiplication yields
\begin{equation*}
    D_\varepsilon^2
    = \begin{bmatrix}
        0 & m - \varepsilon \partial_x \\
        m + \varepsilon \partial_x & 0
    \end{bmatrix}
    \begin{bmatrix}
        0 & m - \varepsilon \partial_x \\
        m + \varepsilon \partial_x & 0
    \end{bmatrix}
    = \begin{bmatrix}
        (m - \varepsilon \partial_x)
            (m + \varepsilon \partial_x) & 0 \\
        0 & (m + \varepsilon \partial_x)
            (m - \varepsilon \partial_x)
    \end{bmatrix}.
\end{equation*}
Applying the commutator relation $[\partial_x, m] = m'(x)$, we evaluate the diagonal entries:
\begin{align*}
    (m - \varepsilon \partial_x)
        (m + \varepsilon \partial_x)
    &= m^2 - \varepsilon m'
        - \varepsilon^2 \partial_x^2, \\
    (m + \varepsilon \partial_x)
        (m - \varepsilon \partial_x)
    &= m^2 + \varepsilon m'
        - \varepsilon^2 \partial_x^2.
\end{align*}
This confirms the form in \eqref{eq:S2_dirac_squared}.
\end{proof}

Propositions~\ref{prop:S2_symmetry} and \ref{prop:S2_squared} show that approximate WKB eigenfunctions of \eqref{eq:S2_dirac} can be constructed by constructing approximate WKB eigenfunctions of semiclassical Schr\"odinger operators with potential functions $m^2$ and subleading terms $\pm \varepsilon m'$. The goal of the present paper is to understand the effect of the subleading terms on the WKB construction and clarify how these terms relate to the existence of the topologically protected ``zero mode'' of \eqref{eq:S2_dirac} \cite{2013FruchartCarpentier}. Our main result is the systematic derivation of a modified Bohr-Sommerfeld quantization rule for approximate WKB eigenfunctions of \eqref{eq:S2_dirac_squared} (Theorem~\ref{subsec:S5_bs}). We find that the subleading terms contribute a half-integer shift of the standard quantization rule and that this shift is exactly what is needed to make the quantization rule consistent with existence of the exact zero mode (Corollary \ref{cor:exact_zero_mode}).

The remainder of this work is structured as follows. In Section \ref{sec:S3_benchmark}, we consider the case where $m(x) = \tanh(x)$, which puts the squared operator into the form of P\"oschl-Teller operators with explicit spectra. In Section \ref{sec:S4_wkb}, we construct approximate WKB eigenpairs for general $m$ valid away from turning points. In Section \ref{sec:S5_turning_points}, we match solutions valid in classically forbidden and allowed regions through the turning points using the Airy asymptotics and derive the modified Bohr-Sommerfeld quantization condition. In Section \ref{sec:S6_confirmation}, we verify that our WKB results agree with the explicit spectra obtained in Section \ref{sec:S3_benchmark}. In Section \ref{sec:S7_spectral} we discuss some implications of our quantization rule, e.g., for the eigenvalue counting function. Finally, in Section \ref{sec:numerical-results}, we numerically compute eigenpairs of \eqref{eq:S2_dirac_squared} and verify excellent agreement with our WKB asymptotics.

\subsection*{Acknowledgements}

This project began as OS's senior thesis project at the University of Minnesota during the 2025-2026 academic year. ABW would like to thank Alexis Drouot for stimulating discussions. ABW's research was supported in part by grant NSF DMS-2406981.


\section{Exactly Solvable Benchmark: Hyperbolic Tangent}
\label{sec:S3_benchmark}

We now consider a mass profile of the form
\begin{equation*}
    m(x) = \tanh x.
\end{equation*}
This profile clearly satisfies the general assumptions \eqref{eq:S2_mass_req}. 

\subsection{Squared Dirac Operator and Large Parameter}
\label{subsec:S3_squared_operator}

\begin{proposition}[P\"{o}schl-Teller Structure]
\label{prop:S3_PT}
For the mass profile $m(x) = \tanh x$, the squared Dirac operator $D_\varepsilon^2$ decomposes as
\begin{equation} \label{eq:S3_dirac_PT_form}
    D_\varepsilon^2
    = 1 + \varepsilon^2
    \begin{pmatrix}
        H_{\varepsilon^{-1}} & 0 \\
        0 & H_{\varepsilon^{-1} - 1}
    \end{pmatrix},
\end{equation}
where $H_\Lambda = - \partial_x^2 - \Lambda (\Lambda + 1) \sech^2 x$ is the P\"{o}schl--Teller operator with parameter $\Lambda$ \cite{Fluegge1971}.
\end{proposition}

\begin{proof}
Using the identities $(\tanh x)' = \sech^2 x$ and $\tanh^2 x = 1 - \sech^2 x$, we obtain, in explicit form, the squared Dirac operator from Equation~\eqref{eq:S2_dirac_squared}:
\begin{equation*}
    D_\varepsilon^2
    = \operatorname{diag} \Big(
        - \varepsilon^2 \partial_x^2
            + 1 - (1 + \varepsilon) \sech^2 x, \;
        - \varepsilon^2 \partial_x^2
            + 1 - (1 - \varepsilon) \sech^2 x
    \Big).
\end{equation*}
Rewriting,
\begin{equation*}
    D_\varepsilon^2
    = 1 + \varepsilon^2
    \begin{pmatrix}
        - \partial_x^2
        - \varepsilon^{-1}
            \left( \varepsilon^{-1} + 1 \right)
            \sech^2 x
        & 0 \\
        0 & - \partial_x^2
            - \left( \varepsilon^{-1} - 1 \right)
                \varepsilon^{-1} \sech^2 x
    \end{pmatrix}.
\end{equation*}
Each diagonal entry has the standard P\"{o}schl--Teller form $H_\Lambda = -\partial_x^2 - \Lambda(\Lambda+1)\sech^2 x$, with $\Lambda = \varepsilon^{-1}$ for the upper component and $\Lambda = \varepsilon^{-1} - 1$ for the lower component. This confirms \eqref{eq:S3_dirac_PT_form}.
\end{proof}

\subsection{Bound States from the P\"{o}schl-Teller Potential}
\label{subsec:S3_bound_states}

\begin{proposition}[Exact Spectrum and Interlacing]
\label{prop:S3_spectrum}
Let $N := \lfloor \varepsilon^{-1} \rfloor$. The discrete eigenvalues of $D_\varepsilon^2$ are given by:
\begin{itemize}
    \item \underline{Upper component:}
        \begin{equation} \label{eq:S3_lambda_plus}
            \lambda_{k}^{(+)} = 1 - (1 - \varepsilon k)^2,
            \quad k = 0, 1, \dots, N.
        \end{equation}
    \item \underline{Lower component:}
        \begin{equation} \label{eq:S3_lambda_minus}
            \lambda_{k}^{(-)} = 1 - (1 - \varepsilon(k+1))^2,
            \quad k = 0, 1, \dots, N-1.
        \end{equation}
\end{itemize}
The spectra interlace such that
\begin{equation} \label{eq:S3_interlacing}
    \lambda_k^{(-)} = \lambda_{k+1}^{(+)}.
\end{equation}
Consequently, the zero-energy mode $\lambda_0^{(+)} = 0$ is simple, while all other discrete eigenvalues are doubly degenerate. Tabulating these eigenvalues:
\begin{align*}
    & \text{\underline{Index}} &
    &\text{\underline{Upper Entry}} &
    &\text{\underline{Lower Entry}}
    \\
    &0 &
    &0 &
    &\varepsilon(2 - \varepsilon)
    \\
    &1 &
    &\varepsilon (2 - \varepsilon) &
    &2\varepsilon (2 - 2\varepsilon)
    \\
    &2 &
    &2\varepsilon (2 - 2\varepsilon) &
    &3\varepsilon (2 - 3\varepsilon)
    \\
    &3 &
    &3\varepsilon (2 - 3\varepsilon) &
    &4\varepsilon (2 - 4\varepsilon)
    \\
    &\vdots &
    &\vdots &
    &\vdots
    \\
    &N-1 &
    &(N-1) \varepsilon
        (2 - (N-1) \varepsilon) &
    &N \varepsilon (2 - N \varepsilon)
    \\
    &N &
    &N \varepsilon (2 - N \varepsilon) &
    &
\end{align*}
In particular, if $\varepsilon^{-1} \in \mathbb{N}$, then $N = \varepsilon^{-1}$ and $\lambda_N = 1$. The total number of distinct eigenvalues of $D_\varepsilon^2$, including the zero energy mode, is
\begin{equation} \label{eq:S3_count_squared}
    \# \sigma_\text{discrete}
    \left( D_\varepsilon^2 \right)
    = N + 1
    = \left \lfloor \varepsilon^{-1} \right \rfloor
        + 1.
\end{equation}
From the anticommutation argument of Section~\ref{sec:S2_setup}, both $+\sqrt{\lambda_k}$ and $-\sqrt{\lambda_k}$ are eigenvalues of $D_\varepsilon$, so the number of distinct discrete eigenvalues of $D_\varepsilon$ is
\begin{equation} \label{eq:S3_count_dirac}
    \# \sigma_\text{discrete}
    \left( D_\varepsilon \right)
    = 2 N + 1
    = 2 \left \lfloor \varepsilon^{-1} \right \rfloor
        + 1.
\end{equation}
Furthermore, the essential spectrum of $D_\varepsilon^2$ begins at $1$, so
\begin{equation*}
    \sigma_{\text{ess}}
        \left( D_\varepsilon^2 \right)
    = [1, \, \infty)
\end{equation*}
and by the symmetry argument,
\begin{equation*}
    \sigma_{\text{ess}}
        \left( D_\varepsilon \right)
    = (-\infty, \, -1] \cup [1, \, \infty).
\end{equation*}
\end{proposition}

\begin{proof}
The eigenvalues of the P\"{o}schl--Teller operator $H_\Lambda$ are known to be~\cite{Fluegge1971}
$E_k = -(\Lambda - k)^2$ for $k \in \{0, 1, \dots, \lfloor \Lambda \rfloor\}$. Applying this to \eqref{eq:S3_dirac_PT_form}:
\begin{equation*}
    \lambda_k^{(+)}
    = 1 + \varepsilon^2 (-(\varepsilon^{-1} - k)^2)
    = 1 - (1 - \varepsilon k)^2.
\end{equation*}
The calculation for $\lambda_k^{(-)}$ is analogous. The interlacing $\lambda_k^{(-)} = \lambda_{k+1}^{(+)}$ follows by direct substitution of indices. The simplicity of the zero mode ($k=0$ in the upper entry) arises because the lower entry's index starts effectively at the first excited state of the upper entry.
\end{proof}
\begin{remark}
The exact spectrum derived in this section serves as a benchmark for the general WKB analysis developed in Sections~\ref{sec:S4_wkb} and~\ref{sec:S5_turning_points}. Proposition~\ref{prop:S6_exact} confirms that the modified Bohr--Sommerfeld condition recovers these eigenvalues exactly for all $k$. 
\end{remark}

\subsection{Spectral Gap Structure and Spacing}
\label{subsec:S3_spectral_gap}

\begin{proposition}[Spectral Spacing and Weyl Scaling]
\label{prop:S3_spacing}
The discrete spectrum of $D_\varepsilon^2$ exhibits two distinct regimes of spacing:
\begin{enumerate}
    \item \underline{Near the gap center ($k \approx 0$):} $\lambda_{k+1} - \lambda_k \approx 2\varepsilon$.
    \item \underline{Near the threshold ($k \approx N$):} $\lambda_{k+1} - \lambda_k \approx 2\varepsilon^2$.
\end{enumerate}
In the semiclassical limit $\varepsilon \to 0$, the total number of bound states for $D_\varepsilon$ scales as $\# \sigma_{\mathrm{disc}}(D_\varepsilon) = \frac{2}{\varepsilon} + \mathcal{O}(1)$.
\end{proposition}

\begin{proof}
Equation~\eqref{eq:S3_count_dirac} gives the exact count $2N+1 \sim 2/\varepsilon$. For the spacing, we work with the upper-component eigenvalues $\lambda_k^{(+)}$ from~\eqref{eq:S3_lambda_plus}, writing $\lambda_k := \lambda_k^{(+)}$ for brevity. Expanding,
\begin{equation*}
    \lambda_k
    = 1 - (1 - \varepsilon k)^2
    = 2\varepsilon k - \varepsilon^2 k^2.
\end{equation*}
For small $k$, $\lambda_k \approx 2 \varepsilon k$, giving nearly uniform spacing $\lambda_{k+1} - \lambda_k \approx 2\varepsilon$. Near the threshold, define $j := N - k$ and $\delta := \varepsilon^{-1} - N < 1$. Then
\begin{equation*}
    1 - \lambda_{N-j}
    = (1 - \varepsilon(N-j))^2
    = \varepsilon^2 (\delta + j)^2,
\end{equation*}
so $\lambda_{N-j} \approx 1 - \varepsilon^2 j^2$, and
\begin{equation*}
    \lambda_{N-(j-1)} - \lambda_{N-j}
    \approx \varepsilon^2 j^2 - \varepsilon^2(j-1)^2
    = \varepsilon^2(2j-1)
    \approx 2\varepsilon^2 j,
\end{equation*}
showing quadratic accumulation near the threshold. In summary:
\begin{align*}
    \text{For $k$ near 0:} \quad
        &\lambda_{k+1} - \lambda_k
            \approx 2 \varepsilon \\
    \text{For $k$ near $N$:} \quad
        &\lambda_{k+1} - \lambda_k
            \approx 2 \varepsilon^2
\end{align*}
implying rough uniform filling within the interior of the gap and quadratic accumulation near the threshold. This will be related to a Weyl-type law for the general case in Section~\ref{sec:S7_spectral}. 
\end{proof}


\section{WKB Approximation for the Squared Dirac Operator}
\label{sec:S4_wkb}

\subsection{Reduction to a Semiclassical Schr\"{o}dinger Equation}
\label{subsec:S4_reduction}

We look to solve the eigenvalue problem $D_\varepsilon^2 \Psi = E^2 \Psi$. Since $D_\varepsilon^2$ is diagonal by Proposition~\ref{prop:S2_squared}, this reduces to solving two independent scalar equations
\begin{equation} \label{eq:S4_schrodinger}
    \left( - \varepsilon^2 \frac{\diff^2}{\diff x^2}
        + m(x)^2 + \sigma_\mathcal{D} m'(x)
    \right) \psi(x)
    = \mathcal{E} \psi(x),
\end{equation}
where $\sigma_\mathcal{D} \in \{\pm 1\}$ is the \emph{pseudo-spin index}, distinguishing the two diagonal entries ($\sigma_\mathcal{D} = -1$ for the upper component and $\sigma_\mathcal{D} = +1$ for the lower), and $\mathcal{E} := E^2$. Although each value of $\sigma_\mathcal{D}$ corresponds to a different component of $\Psi(x)$, we use $\psi(x)$ as a generic notation for an eigenfunction of either, treating both cases simultaneously. Rearranging gives the standard semiclassical form
\begin{equation} \label{eq:S4_semiclassical}
    -\varepsilon^2\psi''(x) = Q(x,\varepsilon)\psi(x),
\end{equation}
with
\begin{equation} \label{eq:S4_Q_def}
    Q(x,\varepsilon)
   := \mathcal{E} - m(x)^2
        - \varepsilon\sigma_\mathcal{D} m'(x).
\end{equation}

\subsection{WKB Ansatz and Phase Expansion}
\label{subsec:S4_ansatz}

\begin{lemma}[WKB Phase Equations]
\label{lem:S4_phase}
With the ansatz $\psi(x) = \exp[iS(x)/\varepsilon]$ and formal phase expansion $S(x) = \sum_{k=0}^\infty \varepsilon^k S_k(x)$, the leading and first-order equations are
\begin{align}
    \mathcal{O}(1): \quad
    &(S_0'(x))^2 - \mathcal{E} + m(x)^2
    = 0, \label{eq:S4_order1} \\
    \mathcal{O}(\varepsilon): \quad
    &2S_0'(x)S_1'(x) - iS_0''(x)
        + \sigma_\mathcal{D} m'(x)
    = 0. \label{eq:S4_ordereps}
\end{align}
\end{lemma}

\begin{proof}
The exponential WKB ansatz is standard in semiclassical analysis; see, e.g.,~\cite{BenderOrszag1978} or~\cite{Zworski2012}. Consider the ansatz
\begin{equation*}
    \psi(x)
    = \exp \left[ \frac{iS(x)}{\varepsilon} \right]
\end{equation*}
with phase expansion
\begin{equation} \label{eq:S4_phase_expansion}
    S(x) = \sum_{k=0}^\infty \varepsilon^k S_k(x).
  \end{equation}
Because the small parameter multiplies the highest derivative of~\eqref{eq:S4_semiclassical}, solutions generically exhibit rapid variation at the scale $\varepsilon$. The exponential ansatz captures precisely this rapid-phase behavior. Computing the derivatives,
\begin{equation*}
    \psi'(x) = \frac{iS'(x)}{\varepsilon}\psi(x),
    \qquad
    \psi''(x) = \left(
        \frac{iS''(x)}{\varepsilon}
        - \frac{(S'(x))^2}{\varepsilon^2}
    \right)\psi(x).
\end{equation*}
Inserting into~\eqref{eq:S4_semiclassical} and dividing by $\psi(x)$,
\begin{equation*}
    (S'(x))^2 - \varepsilon i S''(x)
    = \mathcal{E} - m(x)^2
        - \varepsilon \sigma_\mathcal{D} m'(x).
\end{equation*}
Expanding to $\mathcal{O}(\varepsilon)$ using~\eqref{eq:S4_phase_expansion},
\begin{equation*}
    \Big( \left( S_0'(x) \right)^2
        - \mathcal{E} + m(x)^2
    \Big)
    + \varepsilon \Big( 2 S_0'(x) S_1'(x)
        - i S_0''(x)
        + \sigma_\mathcal{D} m'(x)
    \Big)
    + \mathcal{O} (\varepsilon^2)
    = 0,
\end{equation*}
from which~\eqref{eq:S4_order1} and~\eqref{eq:S4_ordereps} follow by collecting orders. The additional $\sigma_\mathcal{D}m'(x)$ term at $\mathcal{O}(\varepsilon)$ distinguishes this from the standard scalar WKB expansion; the leading term is the same in both cases.
\end{proof}

\begin{lemma}[Leading and Correction Phases]
\label{lem:S4_phases}
The solutions to the phase equations of Lemma~\ref{lem:S4_phase} are
\begin{equation} \label{eq:S4_S0}
    S_0(x)
    = S_0(x_0) + \sigma_0
        \int_{x_0}^x
        \sqrt{\mathcal{E} - m(s)^2}\,\diff s,
    \quad \sigma_0 \in \{\pm 1\},
\end{equation}
and
\begin{equation} \label{eq:S4_S1}
    S_1(x)
    = S_1(x_1)
    - \frac{i}{4}\ln\!\left(
        \frac{\mathcal{E} - m(x_1)^2}
            {\mathcal{E} - m(x)^2}
    \right)
    - \frac{\sigma_\mathcal{D}\sigma_0}{2} I_2(x),
\end{equation}
where $\sigma_0 = +1$ selects right-moving and $\sigma_0 = -1$ selects left-moving waves, and
\begin{equation} \label{eq:S4_I2_def}
    I_2(x)
    = \begin{dcases}
        \arcsin\!\left(\frac{m(x)}{\sqrt{\mathcal{E}}}\right)
        - \arcsin\!\left(\frac{m(x_1)}{\sqrt{\mathcal{E}}}\right)
        & \mathcal{E} \geq m(x)^2, \\
        -i\operatorname{arccosh}\!\left(
            \frac{|m(x)|}{\sqrt{\mathcal{E}}}\right)
        + i\operatorname{arccosh}\!\left(
            \frac{|m(x_1)|}{\sqrt{\mathcal{E}}}\right)
        & \mathcal{E} < m(x)^2.
    \end{dcases}
\end{equation}
\end{lemma}

\begin{proof}
\subsubsection*{Leading Order Phase}

For the $\mathcal{O}(1)$ condition~\eqref{eq:S4_order1},
\begin{equation*}
    S_0'(x)
    = \sigma_0 \sqrt{\mathcal{E} - m(x)^2},
    \quad
    \sigma_0 \in \{ \pm 1 \},
\end{equation*}
where $\sigma_0 = +1$ selects right-moving and $\sigma_0 = -1$ selects left-moving waves. The quantity $\mathcal{E} - m(x)^2$ may be positive or negative: in the classically allowed region
$\sqrt{\mathcal{E} - m(x)^2}$ is taken as the positive real square root, while in the classically forbidden region it is taken as $i\sqrt{m(x)^2 - \mathcal{E}}$, so that $S_0(x)$ becomes purely imaginary there and the WKB solution $e^{iS_0 / \varepsilon}$ decays exponentially. Integrating from a reference point $x_0$ gives~\eqref{eq:S4_S0}.

\subsubsection*{Correction Phase}

For the $\mathcal{O}(\varepsilon)$ condition~\eqref{eq:S4_ordereps},
\begin{equation*}
    S_1'(x)
    = \frac{i S_0''(x) - \sigma_\mathcal{D} m'(x)}
        {2 S_0'(x)}.
\end{equation*}
Differentiating $S_0'(x)$,
\begin{equation*}
    S_0''(x)
    = - \sigma_0
        \frac{m(x) m'(x)}
            {\sqrt{\mathcal{E} - m(x)^2}},
\end{equation*}
which gives
\begin{equation*}
    S_1'(x)
    = \frac{- i m(x) m'(x)}
            {2 \big( \mathcal{E} - m(x)^2 \big)}
        - \frac{ \sigma_\mathcal{D} \sigma_0 m'(x)}
            {2 \sqrt{\mathcal{E} - m(x)^2}}.
\end{equation*}
Note that the second term is entirely due to the additional $\sigma_\mathcal{D} m'(x)$ in the Schr\"{o}dinger-type equation; without it, $S_1'$ would match the standard scalar WKB form. Integrating from a reference point $x_1$,
\begin{equation*}
    S_1(x)
    = S_1(x_1) - \frac{i}{2}
        \underbrace{\int_{x_1}^x
            \frac{m(s) m'(s)}{\mathcal{E} - m(s)^2}
            \, \diff s}_{I_1(x)}
    - \frac{\sigma_\mathcal{D} \sigma_0}{2}
        \underbrace{\int_{x_1}^x
            \frac{m'(s)}{\sqrt{\mathcal{E} - m(s)^2}}
            \, \diff s}_{I_2(x)}.
\end{equation*}
For the first integral $I_1(x)$, the substitution $u = \mathcal{E} - m(s)^2$ with $\diff u = - 2 m'(s) m(s) \, \diff s$ gives
\begin{equation*}
    I_1(x)
    = \frac{1}{2} \ln \left(
        \frac{\mathcal{E} - m(x_1)^2}
        {\mathcal{E} - m(x)^2}
    \right).
\end{equation*}
For the second integral, the substitution $v = m(s)$ gives
\begin{equation*}
    I_2(x)
    = \int_{m(x_1)}^{m(x)}
        \frac{\diff v}{\sqrt{\mathcal{E} - v^2}},
\end{equation*}
which evaluates to~\eqref{eq:S4_I2_def} using the standard antiderivatives. We do not evaluate this integral over the turning points where $\mathcal{E} = m(x)^2$, treating the classically allowed and classically forbidden regions independently. Note that at a classical turning point where $m(x)^2 = \mathcal{E}$, this expression contributes a discrete phase: $\arcsin(\pm 1) = \pm \pi/2$, while $\operatorname{arccosh}(1) = 0$. This foreshadows the modified semiclassical quantization condition derived later. Combining $I_1$ and $I_2$ gives~\eqref{eq:S4_S1}. The reference point $x_1$ remains free at this stage; it will be fixed to $x_1 = x_\tau$ when evaluating the solution near each turning point in Section~\ref{subsec:S5_linearize_wkb}, which simplifies the amplitude terms $A_\chi$ considerably.
\end{proof}

\subsection{Final WKB Form}
\label{subsec:S4_wkb_form}

\begin{proposition}[WKB Solution]
\label{prop:S4_wkb}
Away from turning points, the WKB solution to~\eqref{eq:S4_semiclassical} takes the form
\begin{equation} \label{eq:S4_wkb_final}
    \psi_\chi(x)
    = \frac{CA_\chi P_\chi(x)}
        {\sqrt[4]{\mathcal{E} - m(x)^2}}
    \exp\!\left[
        \frac{i\sigma_0}{\varepsilon}
        \int_{x_0}^x
            \sqrt{\mathcal{E} - m(s)^2}\,\diff s
    \right],
\end{equation}
where $C := \exp[iS_0(x_0)/\varepsilon + iS_1(x_1)]
\sqrt[4]{\mathcal{E} - m(x_1)^2}$, and with
$\chi := \operatorname{sgn}(\mathcal{E} - m(x)^2)$,
\begin{equation} \label{eq:S4_P_A_plus}
    P_+(x)
    = \exp \left[
        - \frac{i \sigma_\mathcal{D} \sigma_0}{2}
        \arcsin \! \left(
            \frac{m(x)}{\sqrt{\mathcal{E}}}
        \right)
    \right], \qquad
    A_+
    = \exp \left[
        \frac{i \sigma_\mathcal{D} \sigma_0}{2}
        \arcsin \! \left(
            \frac{m(x_1)}{\sqrt{\mathcal{E}}}
        \right)
    \right],
\end{equation}
in the classically allowed region ($\chi = +1$), and
\begin{equation} \label{eq:S4_P_A_minus}
    P_-(x)
    = \exp \! \left[
        \frac{\sigma_\mathcal{D} \sigma_0}{2}
        \operatorname{arccosh} \! \left(
            \frac{|m (x)|}{\sqrt{\mathcal{E}}}
        \right)
    \right], \qquad
    A_-
    = \exp \! \left[
        - \frac{\sigma_\mathcal{D} \sigma_0}{2}
        \operatorname{arccosh} \! \left(
            \frac{|m (x_1)|}{\sqrt{\mathcal{E}}}
        \right)
    \right],
\end{equation}
in the classically forbidden region ($\chi = -1$). The solution is valid provided the validity criterion
\begin{equation} \label{eq:S4_wkb_validity}
    \frac{\varepsilon|m(x)m'(x)|}
        {|\mathcal{E} - m(x)^2|^{3/2}}
    \ll 1
\end{equation}
holds, which fails near turning points where $\mathcal{E} - m(x)^2 \to 0$.
\end{proposition}

\begin{proof}
Combining $S_0$ and $S_1$ from Lemma~\ref{lem:S4_phases} into the ansatz gives~\eqref{eq:S4_wkb_final} after collecting the exponential and logarithmic terms. The factor $1/\sqrt[4]{\mathcal{E} - m(x)^2}$ is the standard WKB amplitude arising from the transport equation. The $\sigma_\mathcal{D}$-dependent contribution in $P_\chi$ originates from the $I_2$ term in $S_1$, which is absent in the scalar Schr\"{o}dinger case. Unlike the scalar case, this term produces an additional spin-dependent phase and distinguishes the two branches.

The validity criterion~\eqref{eq:S4_wkb_validity} follows from requiring $|\varepsilon S_1'| \ll |S_0'|$. Near a turning point, the first term of $S_1'$ dominates: from the expression derived in Lemma~\ref{lem:S4_phases}, $S_1'(x) \sim -\frac{im(x)m'(x)}{2(\mathcal{E} - m(x)^2)}$, which diverges as $\mathcal{E} - m(x)^2 \to 0$ and grows faster than the second term $- \frac{ \sigma_\mathcal{D} \sigma_0 m'(x)} {2 \sqrt{\mathcal{E} - m(x)^2}}$. Imposing $|\varepsilon S_1'| \ll |S_0'|$ on this dominant contribution gives
\begin{equation*}
    \varepsilon \cdot
        \frac{|m(x)m'(x)|}{|\mathcal{E} - m(x)^2|}
    \ll \sqrt{|\mathcal{E} - m(x)^2|},
\end{equation*}
which rearranges to~\eqref{eq:S4_wkb_validity}. Note that the second term of $S_1'$ gives the weaker criterion $\varepsilon|m'(x)|/|\mathcal{E} - m(x)^2| \ll 1$, which is implied by~\eqref{eq:S4_wkb_validity} near a turning point. Substituting the linearization $|\mathcal{E} - m(x)^2| \approx |\lambda||x - x_\tau|$ into~\eqref{eq:S4_wkb_validity} shows the criterion reduces to $|x - x_\tau|^{3/2} \gg \varepsilon$, i.e.\ $|x - x_\tau| \gg \varepsilon^{2/3}$, which identifies the natural inner scale at turning points.
\end{proof}


\section{Turning Point Analysis and Connection Formulae}
\label{sec:S5_turning_points}

We work throughout with the Schr\"{o}dinger equation \eqref{eq:S4_semiclassical}--\eqref{eq:S4_Q_def} derived in Section~\ref{sec:S4_wkb}. The WKB solution~\eqref{eq:S4_wkb_final} breaks down when $\mathcal{E} - m(x)^2 = 0$, since the amplitude factor $1/\sqrt[4]{\mathcal{E} - m(x)^2}$ diverges at such points. We call these \emph{turning points}: locations $x_\tau$ satisfying
\begin{equation} \label{eq:S5_turning_point}
    \mathcal{E} - m(x_\tau)^2 = 0,
\end{equation}
which mark the boundaries between the \emph{classically allowed} ($\mathcal{E} - m(x)^2 > 0$) and \emph{classically forbidden} ($\mathcal{E} - m(x)^2 < 0$) regions.

To construct a uniform approximation near each turning point, the real line is decomposed into three zones. The precise boundaries of these zones depend on the natural inner length scale, which is determined in Section~\ref{subsec:S5_linearize_Q} via a dominant balance argument. Anticipating that result (see Lemma~\ref{lem:S5_airy} below), $\Delta = \mathcal{O} (\varepsilon^{2/3})$, we define:
\begin{itemize}
    \item \underline{Outer (WKB) region:}
        $|x - x_\tau| \gg \varepsilon^{2/3}$, where the validity criterion~\eqref{eq:S4_wkb_validity} is satisfied and the WKB approximation holds (see Lemma~\ref{lem:S5_airy} for the derivation of this scale).
    \item \underline{Inner (Airy) region:}
        $|x - x_\tau| \ll 1$, where $Q(x, \varepsilon)$ may be linearized and the resulting equation solved exactly in terms of Airy functions.
    \item \underline{Overlap region:}
        $\varepsilon^{2/3} \ll |x - x_\tau| \ll 1$, where both approximations are simultaneously valid and asymptotic matching between the WKB and Airy solutions is performed.
\end{itemize}
The existence of a non-empty overlap region is guaranteed for $\varepsilon$ sufficiently small, since the inner and outer regions overlap whenever $\varepsilon^{2/3} \ll 1$.

To construct a uniform approximation across a turning point, we proceed in three steps: linearize $Q$ and reduce to Airy's equation (Section~\ref{subsec:S5_linearize_Q}), linearize the WKB solution in the overlap region (Section~\ref{subsec:S5_linearize_wkb}), and match the two representations to determine the connection coefficients (Section~\ref{subsec:S5_connection}).

\subsection{Linearizing the Potential Function}
\label{subsec:S5_linearize_Q}

\begin{lemma}[Airy Reduction]
\label{lem:S5_airy}
Near each turning point $x_\tau$, the equation~\eqref{eq:S4_semiclassical} reduces to the Airy equation
\begin{equation} \label{eq:S5_airy}
    \psi''(\zeta) = \zeta\,\psi(\zeta),
    \qquad
    \zeta := \lambda^{1/3}\varepsilon^{-2/3}(x - x_\tau),
\end{equation}
with errors $\mathcal{O}(\varepsilon + \varepsilon\delta_x + \delta_x^2)$, where
\begin{equation} \label{eq:S5_lambda}
    \lambda := 2m(x_\tau)m'(x_\tau),
\end{equation}
and $\lambda^{1/3}$ denotes the real cube root. In particular,
\begin{equation} \label{eq:S5_linearized_Q}
    Q(x,\varepsilon)
    = -\lambda\delta_x
        + \mathcal{O}\!\left(
            \varepsilon + \varepsilon\delta_x
            + \delta_x^2
        \right),
    \qquad \delta_x := x - x_\tau.
\end{equation}
The natural inner length scale is $\Delta = \mathcal{O}(\varepsilon^{2/3})$, confirming the outer region boundary $|x - x_\tau| \gg \varepsilon^{2/3}$.
\end{lemma}

\begin{proof}
Expanding $Q$ near $x_\tau$ via Taylor series, and using the fact that $x_\tau$ is a root of the $\mathcal{O}(1)$ component of $Q$,
\begin{equation*}
    Q(x, \varepsilon)
    = \underbrace{
        - \varepsilon \sigma_\mathcal{D}
            m'(x_\tau)
        }_{C_{1, \mathcal{O}(\varepsilon)}}
    - \delta_x
    \big( \underbrace{
            2 m(x_\tau) m'(x_\tau)
        }_{\lambda}
        + \underbrace{
            \varepsilon \sigma_\mathcal{D}
                m''(x_\tau)
        }_{C_{2, \mathcal{O}(\varepsilon)}} \big)
    + \mathcal{O} \big( \delta_x^2 \big).
\end{equation*}
Introducing $\zeta = \delta_x/\Delta$ with $\Delta = \mathcal{O}(\varepsilon^\alpha)$ and applying the chain rule, the orders of each coefficient are
\begin{equation*}
    \frac{\varepsilon^2}{\Delta^2}
    = \mathcal{O}(\varepsilon^{2(1-\alpha)}),
    \quad
    \lambda\Delta = \mathcal{O}(\varepsilon^\alpha),
    \quad
    C_{2,\mathcal{O}(\varepsilon)}\Delta
    = \mathcal{O}(\varepsilon^{1+\alpha}),
    \quad
    C_{1,\mathcal{O}(\varepsilon)}
    = \mathcal{O}(\varepsilon).
\end{equation*}
Dominant balance between the left-hand side and the $\lambda\Delta$ term requires
\begin{equation*}
    \mathcal{O}(\varepsilon^{2(1-\alpha)})
    = \mathcal{O}(\varepsilon^\alpha)
    \implies \alpha = \frac{2}{3}
    \implies \Delta = \mathcal{O}(\varepsilon^{2/3}).
\end{equation*}
The remaining terms $C_{2,\mathcal{O}(\varepsilon)}\Delta = \mathcal{O}(\varepsilon^{5/3})$ and $C_{1,\mathcal{O}(\varepsilon)} = \mathcal{O}(\varepsilon)$ are both asymptotically smaller than $\mathcal{O}(\varepsilon^{2/3})$ and may be dropped. Choosing exactly $\Delta = \lambda^{-1/3}\varepsilon^{2/3}$ gives the Airy equation~\eqref{eq:S5_airy}. Since $\lambda^{1/3}$ denotes the real cube root, $\zeta$ is always real. Since $\lambda < 0$ at the left turning point, $\lambda^{1/3} < 0$ there, and since $\lambda > 0$ at the right turning point, $\lambda^{1/3} > 0$ there. The sign of $\zeta$ in each of the four regions is therefore determined by the sign of $\lambda^{1/3}$ and the side of the turning point, and is summarized in Table~\ref{tab:zeta_signs}.
\end{proof}

\begin{table}[ht]
    \centering
    \begin{tabular}{|ll|cc|ccccc|}
        \hline
        Turning Point & Side
            & $\tau$ & $\chi$
            & $\lambda$
                & $\delta_x = x - x_\tau$
            & $\Delta$
            & $\zeta = \delta_x/\Delta$
            & $\zeta^{3/2}$ \\
        \hline
        Left & Forbidden (left)
            & $-1$ & $-1$
            & $< 0$ & $< 0$ & $< 0$ & $> 0$ & real \\
        Left & Allowed (right)
            & $-1$ & $+1$
            & $< 0$ & $> 0$ & $< 0$ & $< 0$ & imaginary \\
        Right & Allowed (left)
            & $+1$ & $+1$
            & $> 0$ & $< 0$ & $> 0$ & $< 0$ & imaginary \\
        Right & Forbidden (right)
            & $+1$ & $-1$
            & $> 0$ & $> 0$ & $> 0$ & $> 0$ & real \\
        \hline
    \end{tabular}
    \caption{Sign of $\zeta$ and character of $\zeta^{3/2}$ in each of the four regions near a turning point, under the nondegenerate turning point assumption of Section~\ref{sec:S2_setup}. The local orientation is determined by the sign of $\partial_x (m^2 - E)$ at $x_\tau$. Here $\Delta = \lambda^{-1/3} \varepsilon^{2/3}$ is a \emph{signed} rescaling factor (not a positive length scale), whose sign is determined by $\operatorname{sgn} (\lambda^{1/3})$. The sign of $\zeta$ determines whether the local solution is oscillatory ($\zeta < 0$, classically allowed) or exponential ($\zeta > 0$, classically forbidden).}
    \label{tab:zeta_signs}
\end{table}

Regardless of which turning point is under consideration, $\zeta < 0$ always corresponds to the classically allowed region and $\zeta > 0$ always corresponds to the classically forbidden region. Consequently, $\zeta^{3/2}$ is purely imaginary in the allowed region, producing oscillatory behavior, and real in the forbidden region, producing exponential decay. This uniform correspondence is what makes the Airy matching procedure tractable.

\subsection{Linearizing the WKB Solution}
\label{subsec:S5_linearize_wkb}

We use the parameter
\begin{equation*}
    \tau := \operatorname{sgn}(m(x_\tau))
    = \begin{cases}
        -1 & \text{left turning point},
            \ m(x_\tau) < 0 \\
        +1 & \text{right turning point},
            \ m(x_\tau) > 0
    \end{cases}
\end{equation*}
to distinguish the two turning points compactly, noting that
\begin{equation} \label{eq:S5_mtp}
    m(x_\tau)
    = \tau \sqrt{\mathcal{E}}.
\end{equation}
In addition, recall the parameter $\chi := \operatorname{sgn}(\mathcal{E} - m(x)^2)$ from Proposition~\ref{prop:S4_wkb}. In the overlap region, $\varepsilon^{2/3} \ll |x - x_\tau| \ll 1$, so
\begin{equation} \label{eq:S5_zeta_bound_above_1}
    |\zeta|
    = |\lambda^{1/3}\varepsilon^{-2/3}(x - x_\tau)|
    \gg 1.
\end{equation}
The linearization of $\mathcal{E} - m(x)^2$ gives
\begin{equation} \label{eq:S5_lin_E_x}
    \mathcal{E} - m(x)^2
    = -\lambda\delta_x
        + \mathcal{O}(\delta_x^2),
\end{equation}
and in the overlap region,
\begin{equation} \label{eq:S5_lin_E_zeta}
    \mathcal{E} - m(x)^2
    = -(\lambda\varepsilon)^{2/3}\zeta
        + \mathcal{O}(\zeta^2\varepsilon^{4/3}).
\end{equation}

\begin{lemma}[Linearized WKB Components]
\label{lem:S5_linearized_wkb}
In the overlap region, the three components of the WKB solution~\eqref{eq:S4_wkb_final} linearize as follows.

\noindent\emph{(i) Phase integral:}
\begin{equation} \label{eq:S5_S0_linearized}
    S_\tau(\zeta)
    = \frac{2i\tau\varepsilon}{3}\zeta^{3/2}
        + \mathcal{O}(\varepsilon^{5/3}\zeta^{5/2}).
\end{equation}

\noindent\emph{(ii) Amplitude terms} (with $x_1 = x_\tau$):
\begin{equation} \label{eq:S5_A_plus_eval}
    A_{+,\tau}
    = \exp\!\left[
        \frac{i\pi\tau\sigma_\mathcal{D}\sigma_0}{4}
    \right], \qquad
    A_{-} = 1.
\end{equation}

\noindent\emph{(iii) Phase functions:}
\begin{equation} \label{eq:S5_P_plus_linearized}
    P_{+,\tau}(\zeta)
    = \exp\!\left[
        -\frac{i\tau\sigma_\mathcal{D}\sigma_0}{2}
        \left(
            \frac{\pi}{2}
            - \frac{(\varepsilon\lambda)^{1/3}}
                {\sqrt{\mathcal{E}}}
                \sqrt{|\zeta|}
        \right)
    \right]
    \exp\!\left[
        \mathcal{O}(\varepsilon|\zeta|^{3/2})
    \right],
\end{equation}
\begin{equation} \label{eq:S5_P_minus_linearized}
    P_{-,\tau}(\zeta)
    = \exp\!\left[
        \frac{\sigma_\mathcal{D}\sigma_0}{2}
        \frac{(\varepsilon\lambda)^{1/3}}
            {\sqrt{\mathcal{E}}}
            \sqrt{\zeta}
    \right]
    \exp\!\left[
        \mathcal{O}(\varepsilon|\zeta|^{3/2})
    \right].
\end{equation}
\end{lemma}

\begin{proof}
\subsubsection*{(i) Phase Integral}

Recall the integral determining $S_0(x)$ from Equation~\eqref{eq:S4_S0}. Near a turning point $x_\tau$, we write $S_\tau(x) := S_0(x)$ to represent this dependence explicitly. Linearizing via~\eqref{eq:S5_lin_E_x} and substituting $u = s - x_\tau$,
\begin{equation*}
    \sqrt{\mathcal{E} - m(s)^2}
    = \sqrt{-\lambda u + \mathcal{O}(u^2)}
    = \sqrt{-\lambda}
        \left( \sqrt{u}
            + \mathcal{O}(u^{3/2})
        \right).
\end{equation*}
Since $x_0$ is arbitrary, we choose $x_0 = x_\tau$ (consistent with the choice $x_1 = x_\tau$ for the correction phase in Lemma~\ref{lem:S4_phases}), and integrate from $x_\tau$ to $x$:
\begin{equation*}
    S_{\tau}(x)
    = \frac{2}{3}\sqrt{-\lambda}\,\delta_x^{3/2}
        + \mathcal{O}(\delta_x^{5/2}).
\end{equation*}
Using $\delta_x^{3/2} = \varepsilon\zeta^{3/2}/\sqrt{\lambda}$ from~\eqref{eq:S5_airy},
\begin{equation*}
    S_{\tau}(\zeta)
    = \frac{2\varepsilon}{3}
        \frac{\sqrt{-\lambda}}{\sqrt{\lambda}}
            \zeta^{3/2}
    + \mathcal{O}(\varepsilon^{5/3}\zeta^{5/2}).
\end{equation*}
The ratio $\sqrt{-\lambda}/\sqrt{\lambda}$ depends on the sign of $\lambda$, which equals the sign of $\tau$:
\begin{equation*}
    \lambda < 0: \
    \frac{\sqrt{-\lambda}}{\sqrt{\lambda}} = -i,
    \qquad
    \lambda > 0: \
    \frac{\sqrt{-\lambda}}{\sqrt{\lambda}} = i,
\end{equation*}
where the principal branch of the square root is assumed. Hence in either case $\sqrt{-\lambda}/\sqrt{\lambda} = \tau i$, giving~\eqref{eq:S5_S0_linearized}.

\subsubsection*{(ii) Amplitude Terms}

Setting $x_1 = x_\tau$ in Equations~\eqref{eq:S4_P_A_plus} and~\eqref{eq:S4_P_A_minus}, and using $m(x_\tau)/\sqrt{\mathcal{E}} = \tau$ from~\eqref{eq:S5_mtp}:
\begin{equation*}
    A_{+, \tau}
    = \exp \left[
        \frac{i \sigma_\mathcal{D} \sigma_0}{2}
        \arcsin(\tau)
    \right]
    = \exp \left[
        \frac{i \pi \tau \sigma_\mathcal{D} \sigma_0}{4}
    \right],
\end{equation*}
using $\arcsin(\pm 1) = \pm \pi/2$. For $A_-$, since $|m(x_\tau)| = \sqrt{\mathcal{E}}$ and $\operatorname{arccosh}(1) = 0$,
\begin{equation*}
    A_{-} = \exp\!\left[
        -\frac{\sigma_\mathcal{D}\sigma_0}{2}
        \operatorname{arccosh}(1)
    \right] = 1,
\end{equation*}
where the subscript $\tau$ is removed since there is no $\tau$ dependence. This gives~\eqref{eq:S5_A_plus_eval}.

\subsubsection*{(iii) Phase Functions}

We first linearize the argument $m(x)/\sqrt{\mathcal{E}}$ around $x_\tau$. Using $m(x_\tau)/\sqrt{\mathcal{E}} = \tau$ and $\delta_x = \lambda^{-1/3}\varepsilon^{2/3}\zeta$ from~\eqref{eq:S5_airy},
\begin{align*}
    \frac{m(x)}{\sqrt{\mathcal{E}}}
    &= \tau
    + \frac{m'(x_\tau)}{\sqrt{\mathcal{E}}}
        \left(
            \frac{\varepsilon^2}{\lambda}
        \right)^{1/3} \zeta
    + \mathcal{O}(\zeta^2\varepsilon^{4/3}).
\end{align*}
Reorganizing the second term using the definition of $\lambda$~\eqref{eq:S5_lambda} and $m(x_\tau)\sqrt{\mathcal{E}} = \tau\mathcal{E}$,
\begin{equation*}
    \frac{m'(x_\tau)}{\sqrt{\mathcal{E}}}
        \left(\frac{\varepsilon^2}{\lambda}\right)^{1/3}
    = \frac{m'(x_\tau)}{\sqrt{\mathcal{E}}}
        \frac{(\varepsilon\lambda)^{2/3}}{\lambda}
    = \frac{m'(x_\tau)}{\sqrt{\mathcal{E}}}
        \frac{(\varepsilon\lambda)^{2/3}}
            {2m(x_\tau)m'(x_\tau)}
    = \frac{\tau(\varepsilon\lambda)^{2/3}}
        {2\mathcal{E}},
\end{equation*}
giving the linearization
\begin{equation} \label{eq:S5_m_linearized}
    \frac{m(x)}{\sqrt{\mathcal{E}}}
    = \tau \left(
        1 + \frac{(\varepsilon\lambda)^{2/3}}
            {2\mathcal{E}} \zeta
    \right)
    + \mathcal{O}(\zeta^2\varepsilon^{4/3}).
\end{equation}

\noindent\textit{Allowed region ($\chi = +1$, $\zeta < 0$):}
With $\zeta < 0$, the second term in~\eqref{eq:S5_m_linearized} is negative, so
\begin{equation*}
    \frac{m(x)}{\sqrt{\mathcal{E}}}
    = \tau\left(
        1 - \frac{(\varepsilon\lambda)^{2/3}}
                {2\mathcal{E}}|\zeta|
    \right)
    + \mathcal{O}(\zeta^2\varepsilon^{4/3}).
\end{equation*}
Since $|\eta| = \frac{(\varepsilon\lambda)^{2/3}}{2\mathcal{E}}|\zeta| \sim |x - x_\tau| \ll 1$, we apply the standard expansion for $0 < y \ll 1$,
\begin{equation*}
    \arcsin(\pm(1-y))
    = \pm\left(\frac{\pi}{2} - \sqrt{2y}\right)
        + \mathcal{O}(y^{3/2}).
\end{equation*}
The $\mathcal{O}(\zeta^2\varepsilon^{4/3})$ error in the argument propagates through $\arcsin$ with derivative
\begin{equation*}
    \frac{\diff}{\diff(\pm(1-y))}\arcsin(\pm(1-y))
    = \frac{1}{\sqrt{1-(1-y)^2}}
    \approx \frac{1}{\sqrt{2y}}
    \sim \varepsilon^{-1/3}|\zeta|^{-1/2},
\end{equation*}
contributing an additional correction of order $\mathcal{O}(\varepsilon|\zeta|^{3/2})$. Since $y^{3/2} = \mathcal{O}(\varepsilon|\zeta|^{3/2})$, both errors are the same order, yielding
\begin{equation*}
    \arcsin\!\left(\frac{m(x)}{\sqrt{\mathcal{E}}}\right)
    = \tau\left(\frac{\pi}{2}
        - \frac{(\varepsilon\lambda)^{1/3}}
            {\sqrt{\mathcal{E}}}|\zeta|^{1/2}
    \right)
    + \mathcal{O}(\varepsilon|\zeta|^{3/2}),
\end{equation*}
which gives~\eqref{eq:S5_P_plus_linearized}.

\noindent\textit{Forbidden region ($\chi = -1$, $\zeta > 0$):}
For the forbidden region we linearize $|m(x)|/\sqrt{\mathcal{E}}$ instead. An analogous calculation using $m'(x_\tau) > 0$ and $\operatorname{sgn}(m(x_\tau)) = \tau$ gives
\begin{equation*}
    \frac{|m(x)|}{\sqrt{\mathcal{E}}}
    = 1 + \frac{(\varepsilon\lambda)^{2/3}}
        {2\mathcal{E}}\zeta
    + \mathcal{O}(\zeta^2\varepsilon^{4/3}).
\end{equation*}
With $\zeta > 0$, $\eta > 0$ and $\eta \ll 1$. Applying the analogous expansion for $0 < y \ll 1$,
\begin{equation*}
    \operatorname{arccosh}(1+y)
    = \sqrt{2y} + \mathcal{O}(y^{3/2}),
\end{equation*}
with the same error propagation argument gives
\begin{equation*}
    \operatorname{arccosh}\!\left(
        \frac{|m(x)|}{\sqrt{\mathcal{E}}}
    \right)
    = \frac{(\varepsilon\lambda)^{1/3}}
        {\sqrt{\mathcal{E}}}\sqrt{\zeta}
    + \mathcal{O}(\varepsilon|\zeta|^{3/2}),
\end{equation*}
which gives~\eqref{eq:S5_P_minus_linearized}.
\end{proof}

\begin{lemma}[WKB in the Overlap Region]
\label{lem:S5_wkb_overlap}
In the overlap region, both cases reduce to
\begin{equation} \label{eq:S5_wkb_overlap}
    \psi_{\chi,\tau}(\zeta)
    = C\!\left(
        \left(-(\lambda\varepsilon)^{2/3}\zeta\right)^{-1/4}
        + \mathcal{O}(\varepsilon^{1/2}\zeta^{3/4})
    \right)
    \exp\!\left[
        -\frac{2\tau\sigma_0}{3}|\zeta|^{3/2}
        + \mathcal{O}\!\left(
            \varepsilon^{1/3}|\zeta|^{1/2}
            + \varepsilon|\zeta|^{3/2}
            + \varepsilon^{5/3}\zeta^{5/2}
        \right)
    \right].
\end{equation}
The full oscillatory solution in the allowed region is
\begin{equation} \label{eq:S5_wkb_full_oscillatory}
    \Psi_{+,\tau}(\zeta)
    = P_\tau\,\psi_{+,\tau}^{(+)}(\zeta)
        + N_\tau\,\psi_{+,\tau}^{(-)}(\zeta),
\end{equation}
where in the classically allowed region ($\chi = +1$, $\zeta < 0$),
\begin{equation} \label{eq:S5_wkb_allowed}
    \psi_{+,\tau}^{(\sigma_0)}(\zeta)
    = C
    \bigg( \left(
        -(\lambda\varepsilon)^{2/3}\zeta
        \right)^{-1/4}
        + \mathcal{O}(\varepsilon^{1/2}\zeta^{3/4})
    \bigg)
    \exp \left[
        i\tau\sigma_0\!\left(
            \frac{\sigma_\mathcal{D}(\varepsilon\lambda)^{1/3}}
                {2\sqrt{\mathcal{E}}}
            |\zeta|^{1/2}
            + \frac{2}{3}|\zeta|^{3/2}
        \right)
        + \mathcal{O}\!\left(
            \varepsilon|\zeta|^{3/2}
            + \varepsilon^{5/3}\zeta^{5/2}
        \right)
    \right],
\end{equation}
and the physical (decaying) solution in the forbidden region (with $\sigma_0 = \tau$) is
\begin{equation} \label{eq:S5_wkb_full_forbidden}
    \Psi_{-,\tau}(\zeta)
    = C\!\left(
        \left(-(\lambda\varepsilon)^{2/3}\zeta\right)^{-1/4}
        + \mathcal{O}(\varepsilon^{1/2}\zeta^{3/4})
    \right)
    \exp\!\left[
        \frac{\tau\sigma_\mathcal{D}
                (\varepsilon\lambda)^{1/3}}
            {2\sqrt{\mathcal{E}}}\zeta^{1/2}
        - \frac{2}{3}\zeta^{3/2}
        + \mathcal{O}(\varepsilon|\zeta|^{3/2}
            + \varepsilon^{5/3}\zeta^{5/2})
    \right].
\end{equation}
\end{lemma}

\begin{proof}
Assembling the components from Lemma~\ref{lem:S5_linearized_wkb}, the amplitude expands as
\begin{equation*}
    (\mathcal{E} - m(x)^2)^{-1/4}
    = \left(-(\lambda\varepsilon)^{2/3}\zeta\right)^{-1/4}
        + \mathcal{O}(\varepsilon^{1/2}\zeta^{3/4}),
\end{equation*}
using $(f + \delta)^{-1/4} \approx f^{-1/4} - \delta/(4f^{5/4})$ with $f = -(\lambda\varepsilon)^{2/3} \zeta$ and $\delta = \mathcal{O} (\varepsilon^{4/3}\zeta^2)$.

\noindent\textit{Classically allowed region ($\chi = +1$, $\zeta < 0$):} Combining with $A_{+,\tau}P_{+,\tau}(\zeta)$ from~\eqref{eq:S5_A_plus_eval} and~\eqref{eq:S5_P_plus_linearized},
\begin{equation*}
    A_{+,\tau}P_{+,\tau}(\zeta)
    = \exp\!\left[
        \frac{i\tau\sigma_\mathcal{D}\sigma_0
            (\varepsilon\lambda)^{1/3}}
            {2\sqrt{\mathcal{E}}}
        \sqrt{|\zeta|}
        + \mathcal{O}(\varepsilon|\zeta|^{3/2})
    \right],
\end{equation*}
where $\zeta^{3/2} = -i|\zeta|^{3/2}$ for $\zeta < 0$ is used in assembling the full exponential, yielding~\eqref{eq:S5_wkb_allowed}. Within the classically allowed region, $P_\tau$ (mnemonic: \emph{positive} direction, $\sigma_0 = +1$) and $N_\tau$ (mnemonic: \emph{negative} direction, $\sigma_0 = -1$) label the coefficients for right-moving and left-moving waves respectively, as in~\eqref{eq:S5_wkb_full_oscillatory}.

\noindent\textit{Classically forbidden region ($\chi = -1$, $\zeta > 0$):} Combining with $A_{-,\tau}P_{-,\tau}(\zeta)$ from~\eqref{eq:S5_A_plus_eval} and~\eqref{eq:S5_P_minus_linearized}, the WKB solution must decay as $|x - x_\tau| \to \infty$ to be physical. The $\zeta^{3/2}$ term is leading in the overlap region since $|\zeta| \gg 1$, so we require $\sigma_0\tau = +1$, fixing $\sigma_0 = \tau$. This gives~\eqref{eq:S5_wkb_full_forbidden}.

\noindent In the overlap region $|\zeta| \gg 1$ from~\eqref{eq:S5_zeta_bound_above_1}. The phase functions $P_{\chi,\tau}$ contribute a term of order $\varepsilon^{1/3}|\zeta|^{1/2}$ against the leading $|\zeta|^{3/2}$ from $S_\tau$, with ratio
\begin{equation*}
    \frac{\varepsilon^{1/3}|\zeta|^{1/2}}{|\zeta|^{3/2}}
    = \frac{\varepsilon^{1/3}}{|\zeta|}
    \ll 1,
\end{equation*}
so the phase functions are subleading and may be dropped at leading order, reducing both cases to~\eqref{eq:S5_wkb_overlap}.
\end{proof}

\subsection{WKB Connection Formula}
\label{subsec:S5_connection}

\begin{lemma}[Connection Coefficients]
\label{lem:S5_connection}
The coefficients $P_\tau$ and $N_\tau$ in~\eqref{eq:S5_wkb_full_oscillatory} are
\begin{equation} \label{eq:S5_connection_coeffs}
    P_\tau = \begin{cases}
        -iC & \tau = -1 \\
        C   & \tau = +1
    \end{cases}, \qquad
    N_\tau = \begin{cases}
        C   & \tau = -1 \\
        -iC & \tau = +1
    \end{cases}.
\end{equation}
\end{lemma}

\begin{proof}
In the overlap regions, the general solution to the Airy equation~\eqref{eq:S5_airy} is a linear combination
\begin{equation*}
    \varphi_\tau(\zeta)
    = a\,\mathrm{Ai}(\zeta) + b\,\mathrm{Bi}(\zeta).
\end{equation*}

\subsubsection*{Classically Forbidden Region ($\zeta \to +\infty$)}

The evanescent region corresponds to the forbidden region. In the limit $\zeta \to +\infty$ (see, e.g.,~\cite[§9.7]{DLMF} or~\cite[§10.4]{AbramowitzStegun1964}),
\begin{equation*}
    \mathrm{Ai}(\zeta)
    \approx \frac{1}{2\sqrt{\pi}\zeta^{1/4}}
        \exp\!\left(-\tfrac{2}{3}\zeta^{3/2}\right),
    \qquad
    \mathrm{Bi}(\zeta)
    \approx \frac{1}{\sqrt{\pi}\zeta^{1/4}}
    \exp\!\left(\tfrac{2}{3}\zeta^{3/2}\right).
\end{equation*}
For a physically acceptable solution, the wavefunction must decay as $|x - x_\tau| \to +\infty$. Since $\mathrm{Bi}(\zeta)$ is unbounded, we must set $b = 0$. Matching $\varphi_\tau$ with~\eqref{eq:S5_wkb_overlap} in the overlap region gives
\begin{equation} \label{eq:S5_a}
    a = \frac{2\sqrt{\pi}C}
        {\sqrt[4]{-(\lambda\varepsilon)^{2/3}}}.
\end{equation}

\subsubsection*{Oscillatory Region ($\zeta \to -\infty$)}

With $b = 0$, we use the Airy asymptotic as $\zeta \to -\infty$:
\begin{equation*}
    \mathrm{Ai}(\zeta)
    \approx \frac{1}{\sqrt{\pi}|\zeta|^{1/4}}
        \sin\!\left(
            \tfrac{2}{3}|\zeta|^{3/2} + \tfrac{\pi}{4}
        \right).
\end{equation*}
Writing $\varphi_\tau$ in terms of complex exponentials with $a$ from~\eqref{eq:S5_a},
\begin{equation} \label{eq:S5_airy_exp}
    \varphi_\tau(\zeta)
    \approx \frac{-iC}
        {\sqrt[4]{-(\lambda\varepsilon)^{2/3}|\zeta|}}
    \left(
        e^{i\pi/4}
            \exp\!\left[\tfrac{2i}{3}|\zeta|^{3/2}\right]
        - e^{-i\pi/4}
            \exp\!\left[-\tfrac{2i}{3}|\zeta|^{3/2}\right]
    \right).
\end{equation}
The general WKB solution~\eqref{eq:S5_wkb_full_oscillatory} in the overlap region via~\eqref{eq:S5_wkb_overlap} is
\begin{equation} \label{eq:S5_wkb_oscillatory}
    \Psi_{+,\tau}(\zeta)
    \approx
    \frac{1}
        {\sqrt[4]{(\lambda\varepsilon)^{2/3}|\zeta|}}
    \left(
        P_\tau \exp\!\left[
                -\frac{2i\tau}{3}|\zeta|^{3/2}
            \right]
        + N_\tau \exp\!\left[
                \frac{2i\tau}{3}|\zeta|^{3/2}
            \right]
    \right),
\end{equation}
where $\zeta = -|\zeta|$ is used in the denominator. Using
\begin{equation} \label{eq:S5_ratio}
    \frac{(\lambda\varepsilon)^{1/6}}
        {\left(-(\lambda\varepsilon)^{2/3}\right)^{1/4}}
    = e^{-i\pi/4},
\end{equation}
we now equate coefficients of the exponentials for both turning points.

\noindent Left turning point ($\tau = -1$): The first (second) term of~\eqref{eq:S5_wkb_oscillatory} matches the first (second) term of~\eqref{eq:S5_airy_exp}, giving
\begin{equation*}
    P_- = -iC,
    \qquad
    N_- = C.
\end{equation*}
\noindent Right turning point ($\tau = +1$): The roles of $P$ and $N$ are exchanged, giving
\begin{equation*}
    N_+ = -iC,
    \qquad
    P_+ = C.
\end{equation*}
These yield~\eqref{eq:S5_connection_coeffs}.
\end{proof}

\subsection{Modified Bohr-Sommerfeld Quantization Condition}
\label{subsec:S5_bs}

\begin{theorem}[Modified Bohr-Sommerfeld Condition]
\label{thm:S5_bs}
Under the assumptions of Section~\ref{sec:S2_setup}, with $m'(x_\tau) > 0$ and two simple turning points $x_- < x_+$, global consistency of the WKB solution holds under the following condition on $\mathcal{E}$
\begin{equation} \label{eq:S5_bs}
    \frac{1}{\varepsilon}
    \int_{x_-}^{x_+}
        \sqrt{\mathcal{E} - m(x)^2}\,\diff x
    = \left(k + \frac{\sigma_\mathcal{D}+1}{2}\right)\pi
        + \mathcal{O}(\varepsilon),
    \quad k = 0, 1, 2, \ldots.
\end{equation}
In particular,
\begin{equation*}
    \frac{1}{\varepsilon}
    \int_{x_-}^{x_+}
        \sqrt{\mathcal{E} - m(x)^2}\,\diff x
    + \mathcal{O}(\varepsilon)
    = \begin{cases}
        k\pi & \sigma_\mathcal{D} = -1 \\
        (k+1)\pi & \sigma_\mathcal{D} = +1
    \end{cases}
\end{equation*}
\end{theorem}
\begin{proof}
Having established the connection coefficients~\eqref{eq:S5_connection_coeffs}, we impose global consistency of the WKB solution~\eqref{eq:S5_wkb_full_oscillatory} across the full classically allowed region $x_- < x < x_+$. The key insight is that the $|\zeta|^{1/2}$ corrections cancel identically when the global and local WKB representations are matched, so the full $\sigma_\mathcal{D}$-dependent phase is retained.

\subsubsection*{Global WKB Solution}

The general WKB solution valid throughout the classically allowed region is a superposition of the right-moving ($\sigma_0 = +1$) and left-moving ($\sigma_0 = -1$) branches of~\eqref{eq:S4_wkb_final}:
\begin{equation} \label{eq:S5_global_wkb}
    \Psi_+(x)
    = \frac{1}{(\mathcal{E} - m(x)^2)^{1/4}}
    \left[
        B_R\exp\!\left(
            i\phi_L(x)
            - \frac{i\sigma_\mathcal{D}}{2}\rho(x)
        \right)
        + B_L\exp\!\left(
            -i\phi_L(x)
            + \frac{i\sigma_\mathcal{D}}{2}\rho(x)
        \right)
    \right],
\end{equation}
where $B_R, B_L \in \mathbb{C}$ are the amplitudes for right and left moving waves respectively, $\sigma_0$ is absorbed into the definition of the two amplitudes, and
\begin{equation} \label{eq:S5_phi_rho_def}
    \phi_L(x)
   := \frac{1}{\varepsilon}
        \int_{x_{-}}^x
        \sqrt{\mathcal{E} - m(s)^2} \, \diff s,
    \qquad
    \rho(x)
    := \arcsin\!\left(
        \frac{m(x)}{\sqrt{\mathcal{E}}}
    \right).
\end{equation}
Note that $\phi_L(x_-) = 0$ and $\phi_L(x_+) = \Phi$, where
\begin{equation} \label{eq:S5_Phi}
    \Phi
    := \frac{1}{\varepsilon}
        \int_{x_-}^{x_+}
        \sqrt{\mathcal{E} - m(s)^2} \, \diff s
    \geq 0
\end{equation}
is the total accumulated phase across the well. This expression is valid in the classically allowed region, away from both turning points: for $x$ satisfying
\begin{equation*}
    x_- + \mathcal{O}(\varepsilon^{2/3})
    < x <
    x_+ - \mathcal{O}(\varepsilon^{2/3}).
\end{equation*}
Since both $\Psi_+(x)$ and $\Psi_{+,\tau}(\zeta)$ are valid in the overlap region, we linearize $\Psi_+(x)$ near $x_\tau$ and match it to $\Psi_{+,\tau}(\zeta)$.

\subsubsection*{Linearization in the Overlap Region}

Using Lemma~\ref{lem:S5_linearized_wkb}(i), near each $x_\tau$,
\begin{equation} \label{eq:S5_phi_rho_lin}
    \phi_L(x)
    = \phi_L(x_\tau)
        - \frac{2\tau}{3}|\zeta|^{3/2}
        + \mathcal{O}(\varepsilon^{5/3}\zeta^{5/2}),
    \qquad
    \rho(x)
    = \tau\!\left(
        \frac{\pi}{2}
        - \frac{(\varepsilon\lambda)^{1/3}}
            {\sqrt{\mathcal{E}}}|\zeta|^{1/2}
    \right)
    + \mathcal{O}(\varepsilon|\zeta|^{3/2}),
\end{equation}
where the linearization of $\phi_L$ near $x_-$ gives $\phi_L(x_- + \delta) = \frac{2}{3}|\zeta|^{3/2} + \mathcal{O}(\varepsilon^{5/3}\zeta^{5/2})$ (so $\phi_L(x_-) = 0$), and near $x_+$ gives $\phi_L(x_+ - \delta) = \Phi - \frac{2}{3}|\zeta|^{3/2} + \mathcal{O}(\varepsilon^{5/3}\zeta^{5/2})$ (so $\phi_L(x_+) = \Phi$), and the linearization of $\rho$ follows Lemma~\ref{lem:S5_linearized_wkb}(iii). Substituting into~\eqref{eq:S5_global_wkb}, the global solution near $x_\tau$ becomes
\begin{align} \label{eq:S5_global_near_tp}
    \Psi_+(x)
    &= \frac{1}{(\mathcal{E}-m(x)^2)^{1/4}}
    \Bigg[
        B_R\exp\!\left(
            i\phi_L(x_\tau)
            - \frac{2i\tau}{3}|\zeta|^{3/2}
            - \frac{i\pi\tau\sigma_\mathcal{D}}{4}
            + \frac{i\tau\sigma_\mathcal{D}
                (\varepsilon\lambda)^{1/3}}
                {2\sqrt{\mathcal{E}}}
            |\zeta|^{1/2}
        \right)
        \nonumber \\
        &\hspace{30pt}
        + B_L\exp\!\left(
            -i\phi_L(x_\tau)
            + \frac{2i\tau}{3}|\zeta|^{3/2}
            + \frac{i\pi\tau\sigma_\mathcal{D}}{4}
            - \frac{i\tau\sigma_\mathcal{D}
                (\varepsilon\lambda)^{1/3}}
                {2\sqrt{\mathcal{E}}}
            |\zeta|^{1/2}
        \right)
    \Bigg]
    \exp\!\left[
        \mathcal{O}\!\left(
            \varepsilon|\zeta|^{3/2}
            + \varepsilon^{5/3}\zeta^{5/2}
        \right)
    \right].
\end{align}
The connection formula from Lemma~\ref{lem:S5_connection} and Lemma~\ref{lem:S5_wkb_overlap} gives near $x_\tau$,
\begin{align} \label{eq:S5_connection_near_tp}
    \Psi_{+,\tau}(x)
    &=
    \bigg( \left(
        -(\lambda\varepsilon)^{2/3}\zeta
        \right)^{-1/4}
        + \mathcal{O}(\varepsilon^{1/2}\zeta^{3/4})
    \bigg)
    \Bigg[
        P_\tau\exp\!\left(
            \frac{2i\tau}{3}|\zeta|^{3/2}
            + \frac{i\tau\sigma_\mathcal{D}
                (\varepsilon\lambda)^{1/3}}
                    {2\sqrt{\mathcal{E}}}
                |\zeta|^{1/2}
        \right)
        \nonumber \\
        &\hspace{30pt}
        + N_\tau\exp\!\left(
            -\frac{2i\tau}{3}|\zeta|^{3/2}
            - \frac{i\tau\sigma_\mathcal{D}
                (\varepsilon\lambda)^{1/3}}
                    {2\sqrt{\mathcal{E}}}
                |\zeta|^{1/2}
        \right)
    \Bigg]
    \exp\!\left[
        \mathcal{O}\!\left(
            \varepsilon|\zeta|^{3/2}
            + \varepsilon^{5/3}\zeta^{5/2}
        \right)
    \right].
\end{align}

\subsubsection*{Matching}

The $|\zeta|^{1/2}$ corrections appear in both~\eqref{eq:S5_global_near_tp} and~\eqref{eq:S5_connection_near_tp} with identical coefficient $\frac{i \tau \sigma_\mathcal{D} (\varepsilon \lambda)^{1/3}} {2 \sqrt{\mathcal{E}}} |\zeta|^{1/2}$; since this factor appears in every exponential term on both sides, it factors out and cancels identically when coefficients of $\exp( \pm \frac{2i}{3} |\zeta|^{3/2})$ are equated in the matching step below. The amplitude error $\mathcal{O} (\varepsilon^{1/2} \zeta^{3/4})$ is a common multiplicative factor on both sides and cancels in the ratio $B_L/B_R$. Setting $|\zeta| \sim \varepsilon^\nu$ for $0 < \nu < 2/3$, the error terms satisfy
\begin{equation*}
    \mathcal{O}\!\left(
        \varepsilon|\zeta|^{3/2}
        + \varepsilon^{5/3}|\zeta|^{5/2}
    \right)
    = \mathcal{O}(\varepsilon^{1+3\nu/2}).
\end{equation*}
Equating coefficients of $\exp(\pm\frac{2i}{3}|\zeta|^{3/2})$ gives
\begin{equation} \label{eq:S5_matching_BR}
    B_R\exp\!\left[
        i\phi_L(x_\tau)
        - \frac{i\pi\tau\sigma_\mathcal{D}}{4}
    \right]
    = N_\tau\left(
        1 + \mathcal{O}(\varepsilon^{1+3\nu/2})
    \right),
\end{equation}
\begin{equation} \label{eq:S5_matching_BL}
    B_L\exp\!\left[
        -i\phi_L(x_\tau)
        + \frac{i\pi\tau\sigma_\mathcal{D}}{4}
    \right]
    = P_\tau\left(
        1 + \mathcal{O}(\varepsilon^{1+3\nu/2})
    \right).
\end{equation}

\subsubsection*{Consistency Condition and Quantization}

Dividing~\eqref{eq:S5_matching_BL} by~\eqref{eq:S5_matching_BR},
\begin{equation} \label{eq:S5_ratio_general}
    \frac{B_L}{B_R}
    = \frac{P_\tau}{N_\tau}
        \exp\!\left[
            2i\phi_L(x_\tau)
            - \frac{i\pi\tau\sigma_\mathcal{D}}{2}
        \right]
    \left(1 + \mathcal{O}(\varepsilon^{1+3\nu/2})\right).
\end{equation}
From~\eqref{eq:S5_connection_coeffs}, $P_\tau/N_\tau = i\tau$. Evaluating at each turning point:

\noindent Left turning point: ($\tau = -1$, $\phi_L(x_-) = 0$):
\begin{equation} \label{eq:S5_ratio_left}
    \frac{B_L}{B_R}
    = -i\exp\!\left[
        \frac{i\pi\sigma_\mathcal{D}}{2}
    \right]
    \left(1 + \mathcal{O}(\varepsilon^{1+3\nu/2})\right).
\end{equation}
\noindent Right turning point: ($\tau = +1$, $\phi_L(x_+) = \Phi$):
\begin{equation} \label{eq:S5_ratio_right}
    \frac{B_L}{B_R}
    = i\exp\!\left[
        2i\Phi - \frac{i\pi\sigma_\mathcal{D}}{2}
    \right]
    \left(1 + \mathcal{O}(\varepsilon^{1+3\nu/2})\right).
\end{equation}
Equating~\eqref{eq:S5_ratio_left} and~\eqref{eq:S5_ratio_right} and using $-1 = e^{i\pi(2k+1)}$,
\begin{equation*}
    \Phi = \left(
        k + \frac{\sigma_\mathcal{D}+1}{2}
    \right)\pi
    + \mathcal{O}(\varepsilon^{1+3\nu/2}).
\end{equation*}
Since $\Phi \geq 0$, we restrict $k \geq 0$. Since $0 < \nu < 2/3$, we have $1 + 3 \nu/2 \in (1, 2)$, so the error is $\mathcal{O}(\varepsilon^{1+3\nu/2})$ for any fixed $\nu \in (0, 2/3)$. The uniform $\mathcal{O}(\varepsilon)$ bound stated in~\eqref{eq:S5_bs} corresponds to taking $\nu \to 0^+$, which gives the weakest (most conservative) estimate; choosing $\nu$ closer to $2/3$ would yield $\mathcal{O}(\varepsilon^{2-})$ but requires $|\zeta|$ to be taken exponentially close to $1$, making the overlap region correspondingly thin. The
$\mathcal{O}(\varepsilon)$ bound is therefore uniform and sufficient for the leading-order quantization condition. The condition is asymptotically valid as $\varepsilon \to 0$, with relative errors of order $\mathcal{O}(\varepsilon)$. The confirmation against the P\"{o}schl--Teller benchmark in Section~\ref{sec:S6_confirmation} is therefore nontrivial: for that potential the condition is exact for all $k$. 
\end{proof}

\begin{corollary}[Zero-energy mode via $D_\varepsilon^2$] \label{cor:exact_zero_mode}
Under the assumptions of Section~\ref{sec:S2_setup}, suppose $m$ has a single
simple zero $x_0$ (i.e.\ $m(x_0)=0$, $m'(x_0)\neq 0$) and that
$m(x)\to c_{\pm}>0$ as $x\to\pm\infty$ with the same sign on both sides.
Then $D_\varepsilon^2$ has a zero eigenvalue, and consequently $D_\varepsilon$
has an eigenvalue at $E=0$.
\end{corollary}

\begin{proof}
\textbf{Bohr--Sommerfeld estimate of the lowest eigenvalue.}
Consider the Bohr--Sommerfeld quantization condition
\eqref{eq:S5_bs} in the sector $\sigma_{\mathcal{D}} = -1$. For $\mathcal{E} > 0$, the equation $m(x)^2 = \mathcal{E}$ has two turning points $x_- < x_+$ (which bound $x_0$), and the lowest eigenvalue corresponds to $k = 0$:
\begin{equation} \label{eq:bs_k0}
    \int_{x_-}^{x_+}
        \sqrt{\mathcal{E} - m(x)^2} \, \diff x
    = \mathcal{O}(\varepsilon).
\end{equation}
Because $m$ has a \emph{single} simple zero at $x_0$, the Taylor expansion gives
\begin{equation*}
    m(x)
    = m'(x_0) (x-x_0)
        + \mathcal{O} \! \left((x-x_0)^2\right),
\end{equation*}
and squaring,
\begin{equation*}
    m(x)^2
    = m'(x_0)^2 (x-x_0)^2
        + \mathcal{O} \! \left((x-x_0)^3\right).
\end{equation*}
For small $\mathcal{E}$ the turning points are close to $x_0$, so setting $\mu := |m'(x_0)| > 0$ and $u = x-x_0$ the leading-order turning points satisfy $\mu^2 u_{\pm}^2 = \mathcal{E}$, i.e. \ $u_\pm = \pm\sqrt{\mathcal{E}} / \mu$. Substituting
$u = (\sqrt{\mathcal{E}} / \mu) \sin\theta$ in the leading-order integral:
\begin{align*}
    \int_{-\sqrt{\mathcal{E}}/\mu}
        ^{\sqrt{\mathcal{E}}/\mu}
        \sqrt{\mathcal{E} - \mu^2 u^2} \, \diff u
    &= \frac{\mathcal{E}}{\mu}
        \int_{-\pi/2}^{\pi/2}
        \cos^2\theta \, \diff \theta
    = \frac{\pi \, \mathcal{E}}{2\mu}.
\end{align*}
The error incurred by replacing $m(x)^2$ by its leading-order term is $\mathcal{O}(\mathcal{E}^{3/2})$ uniformly in a neighborhood of $x_0$, so the full integral satisfies
\begin{equation*}
    \int_{x_-}^{x_+}
        \sqrt{\mathcal{E} - m(x)^2} \, \diff x
    = \frac{\pi \, \mathcal{E}}{2|m'(x_0)|}
      + \mathcal{O} \! \left(\mathcal{E}^{3/2}\right).
\end{equation*}
Inserting this into \eqref{eq:bs_k0} and solving for the lowest eigenvalue $\mathcal{E}_0$ yields
\begin{equation*}
    \mathcal{E}_0 = \mathcal{O}(\varepsilon),
\end{equation*}
so the Bohr--Sommerfeld approximation predicts a near-zero mode.

\smallskip
\noindent\textbf{Exact zero mode via supersymmetric
factorisation.}
The Bohr--Sommerfeld estimate shows $\mathcal{E}_0 \to 0$ but does not by itself establish an \emph{exact} zero eigenvalue. We now supply that argument directly. Write $D_\varepsilon^2$ in its block-diagonal form:
\begin{equation*}
    D_\varepsilon^2 = 
    \begin{pmatrix}
        H_+ & 0 \\
        0 & H_-
    \end{pmatrix},
    \qquad
    H_{\pm}
    := -\varepsilon^2\partial_x^2
        + m(x)^2 \mp \varepsilon \, m'(x),
\end{equation*}
as follows from the anticommutation
relations (see Section~\ref{sec:S2_setup}). Define
\begin{equation*}
    M(x)
    := \int_{x_0}^{x} m(t) \, \diff t,
    \qquad
    \psi_0(x)
   := \exp \! \left(
        -\frac{M(x)}{\varepsilon}
    \right).
\end{equation*}
A direct computation gives
\begin{equation*}
    H_- \, \psi_0
    = \left(
        - \varepsilon^2\partial_x^2 + m^2
        - \varepsilon m'
    \right) \psi_0
    = \bigl(
        \varepsilon m'(x) - m(x)^2 + m(x)^2
        - \varepsilon m'(x)
    \bigr) \psi_0 = 0.
  \end{equation*}
It remains to verify that $\psi_0 \in L^2(\mathbb{R})$.  Since $m$ has a single simple zero at $x_0$ and $m(x)\to c_{\pm}>0$ as $x \to \pm \infty$ (same sign on both sides), the antiderivative $M(x)$ satisfies $M(x) \to +\infty$ as $x \to \pm \infty$. Hence $\psi_0(x) \to 0$ exponentially in both directions, so $\psi_0 \in L^2(\mathbb{R}) \setminus \{0\}$.

The function $(0, \psi_0)^{T}$ is a nonzero element of $\ker D_\varepsilon^2$, so $0$ is an eigenvalue of $D_\varepsilon^2$. Since $D_\varepsilon$ is self-adjoint, $\ker D_\varepsilon^2 = \ker D_\varepsilon$, so $D_\varepsilon(0, \psi_0)^{T} = 0$ and $E=0$ is an eigenvalue of $D_\varepsilon$.
\end{proof}

\section{Confirmation Against the Benchmark Case}
\label{sec:S6_confirmation}

\begin{proposition}[Exact Recovery for $m(x) = \tanh x$]
\label{prop:S6_exact}
For $m(x) = \tanh x$, the modified Bohr-Sommerfeld condition~\eqref{eq:S5_bs} recovers the exact P\"{o}schl--Teller eigenvalues of Proposition~\ref{prop:S3_spectrum} for all $k$, with no asymptotic error.
\end{proposition}

\begin{proof}
With $m(x) = \tanh(x)$ and eigenvalue $\mathcal{E}$, the turning points $x_- < 0 < x_+$ satisfying~\eqref{eq:S5_turning_point} are
\begin{equation*}
    x_- = -\operatorname{arctanh}\!\left(\sqrt{\mathcal{E}}\right),
    \qquad
    x_+ = +\operatorname{arctanh}\!\left(\sqrt{\mathcal{E}}\right),
\end{equation*}
since $\tanh^2(x_\tau) = \mathcal{E}$ and $\operatorname{arctanh}$ is odd. The quantization integral~\eqref{eq:S5_Phi} is therefore
\begin{equation*}
    \Phi(\mathcal{E})
    = \frac{1}{\varepsilon}
        \int_{-\operatorname{arctanh}(\sqrt{\mathcal{E}})}
            ^{+\operatorname{arctanh}(\sqrt{\mathcal{E}})}
            \sqrt{\mathcal{E} - \tanh^2(x)}\,\diff x.
\end{equation*}
The substitution $u = \tanh x$, $\diff u = (1-u^2)\diff x$, with limits $u = \pm\sqrt{\mathcal{E}}$, gives
\begin{equation*}
    \Phi(\mathcal{E})
    = \frac{1}{\varepsilon}
    \int_{-\sqrt{\mathcal{E}}}^{\sqrt{\mathcal{E}}}
        \frac{\sqrt{\mathcal{E} - u^2}}{1 - u^2}\,\diff u.
\end{equation*}
The further substitution $u = \sqrt{\mathcal{E}}\sin\theta$, $\diff u = \sqrt{\mathcal{E}}\cos\theta \, \diff\theta$, with limits $\theta = \pm\pi/2$, gives
\begin{equation*}
    \Phi(\mathcal{E})
    = \frac{\mathcal{E}}{\varepsilon}
        \int_{-\pi/2}^{\pi/2}
        \frac{\cos^2\theta}
            {1 - \mathcal{E}\sin^2\theta}\,\diff\theta
    = \frac{2\mathcal{E}}{\varepsilon}
        \int_0^{\pi/2}
        \frac{\cos^2\theta}
            {1 - \mathcal{E}\sin^2\theta}\,\diff\theta,
\end{equation*}
where symmetry of the integrand is used. Applying the closed form~\cite[3.621.3]{GradshteynRyzhik2007}
\begin{equation*}
    \int_0^{\pi/2}
        \frac{\cos^2\theta}{1 - k\sin^2\theta}\,\diff\theta
    = \frac{\pi}{2k}\left(1 - \sqrt{1-k}\right),
    \quad 0 < k < 1,
\end{equation*}
with $k = \mathcal{E} \in (0,1)$,
\begin{equation*}
    \Phi(\mathcal{E})
    = \frac{\pi}{\varepsilon}
        \left(1 - \sqrt{1 - \mathcal{E}}\right).
\end{equation*}
Substituting into the modified Bohr-Sommerfeld condition~\eqref{eq:S5_bs} and solving for $\mathcal{E}$,
\begin{equation*}
    \mathcal{E}_n^{(\sigma_\mathcal{D})}
    = 1 - \left[
        1 - \varepsilon\!\left(
            n + \frac{\sigma_\mathcal{D}+1}{2}
        \right)
    \right]^2.
\end{equation*}
For the upper component $\sigma_\mathcal{D} = -1$,
\begin{equation*}
    \mathcal{E}_n^{(-1)}
    = 1 - (1 - \varepsilon n)^2
    = \lambda_n^{(+)},
\end{equation*}
and for the lower component $\sigma_\mathcal{D} = +1$,
\begin{equation*}
    \mathcal{E}_n^{(+1)}
    = 1 - \big(1 - \varepsilon(n+1)\big)^2
    = \lambda_n^{(-)}.
\end{equation*}
These match the exact P\"{o}schl--Teller eigenvalues~\eqref{eq:S3_lambda_plus} and~\eqref{eq:S3_lambda_minus} precisely. The agreement is exact for all $n$, not merely asymptotic in $\varepsilon$, confirming both the correctness of the modified quantization condition~\eqref{eq:S5_bs} and the one-level offset between components encoded in the $(\sigma_\mathcal{D}+1)/2$ term. In particular, setting $n = 0$ and $\sigma_\mathcal{D} = -1$ gives $\mathcal{E} = 0$, recovering the topologically protected zero mode exclusively from the upper component, consistent with $\lambda_0^{(+)} = 0$ and the interlacing relation~\eqref{eq:S3_interlacing}.
\end{proof}

\section{Spectral Consequences and Weyl Law}
\label{sec:S7_spectral}

Formally, the modified Bohr-Sommerfeld condition~\eqref{eq:S5_bs} determines the eigenvalues $\mathcal{E}_n$ of $D_\varepsilon^2$. Recalling from Section~\ref{sec:S2_setup} that non-zero eigenvalues of $D_\varepsilon^2$ are two-fold degenerate with corresponding eigenvalues $\pm\sqrt{\mathcal{E}_n}$ of $D_\varepsilon$, the quantization condition~\eqref{eq:S5_bs} simultaneously determines the positive and negative branches of the spectrum of $D_\varepsilon$.

\subsection{Weyl-Type Counting Law}
\label{subsec:S7_weyl}

\begin{corollary}[Weyl-Type Counting Law]
\label{cor:S7_weyl}
In the semiclassical limit $\varepsilon \to 0$, Theorem~\ref{thm:S5_bs} suggests the following asymptotics for the eigenvalue counting function $N(\mathcal{E})$ for $D_{\varepsilon}^2$
\begin{equation} \label{eq:S7_weyl_D2}
    N(\mathcal{E})
    \sim \frac{1}{\pi\varepsilon}
    \int_{\{x\,:\,m(x)^2 < \mathcal{E}\}}
        \sqrt{\mathcal{E} - m(x)^2}\,\diff x,
\end{equation}
and, substituting $\mathcal{E} = \mu^2$, the following for the positive eigenvalue counting function for $D_\varepsilon$
\begin{equation} \label{eq:S7_weyl_D}
    N(\mu)
    \sim \frac{1}{\pi\varepsilon}
    \int_{\{x\,:\,m(x)^2 < \mu^2\}}
        \sqrt{\mu^2 - m(x)^2}\,\diff x.
\end{equation}
By the spectral symmetry established in Section~\ref{sec:S2_setup}, the total number of eigenvalues with $|\mu| \leq L$ would then be approximately $2N(L)$, not counting any zero mode.
\end{corollary}

\begin{proof}
Solving Theorem~\ref{thm:S5_bs} for $n$ gives
\begin{equation} \label{eq:S7_n_of_E}
    n(\mathcal{E})
    = \frac{1}{\pi\varepsilon}
    \int_{x_-}^{x_+}
        \sqrt{\mathcal{E} - m(x)^2}
    \, \diff x
    - \frac{\sigma_\mathcal{D}+1}{2}
    + \mathcal{O}(\varepsilon).
\end{equation}
In the limit $\varepsilon \to 0$, the subleading constant and error are negligible, and $n(\mathcal{E})$ approximates the number of eigenvalues below $\mathcal{E}$, giving~\eqref{eq:S7_weyl_D2}. The result~\eqref{eq:S7_weyl_D} follows by substituting $\mathcal{E} = \mu^2$.
\end{proof}

\subsection{Semiclassical Interpretation}
\label{subsec:S7_interpretation}

The integrand $p(x) = \sqrt{\mathcal{E} - m(x)^2}$ plays the role of a local momentum, so the counting function is proportional to the classical phase space volume
\begin{equation*}
    \int_{\{m(x)^2 < \mathcal{E}\}} p(x) \, \diff x.
\end{equation*}
This is the natural analogue of the Weyl law for the Schr\"{o}dinger operator~\cite{DimassiSjostrand1999, Zworski2012}, with $m(x)^2$ replacing the potential $V(x)$. The $\sigma_\mathcal{D}$-dependent shift in~\eqref{eq:S7_n_of_E} is subleading in $\varepsilon$ but encodes a spectral asymmetry between the two components of $D_\varepsilon^2$: the
$\sigma_\mathcal{D} = -1$ condition admits $n = 0$ with $\mathcal{E} = 0$, while the $\sigma_\mathcal{D} = +1$ condition does not. This one-level offset is the semiclassical manifestation of the topologically protected zero mode established in Section~\ref{sec:S2_setup}.

\section{Numerical Results}
\label{sec:numerical-results}

The preceding sections establish the modified Bohr--Sommerfeld condition~\eqref{eq:S5_bs} analytically and confirm it exactly against the P\"oschl--Teller benchmark. We now verify the condition numerically across three mass profiles: the hyperbolic tangent $m(x) = \tanh (x)$, the error function $m(x) = \operatorname{erf} (x)$, and the algebraic sigmoid $m(x) = x / \sqrt{1 + x^2 / 3}$. The first admits an exact closed-form spectrum, providing a precise benchmark; the latter two do not: the quantization integral for $\operatorname{erf}$ has no closed form, and for the algebraic sigmoid it reduces to an elliptic integral, so numerical verification is essential.

Eigenvalues are computed by discretizing the squared operator $D_\varepsilon^2$ on a uniform grid of $N$ points over $[-12, \ 12 ]$, approximating the second derivative via a five-point finite difference stencil with coefficients $5/2$, $-4/3$, $1/12$ (divided by $h^2$). The grid size adapts to the semiclassical parameter: a minimum of $N=1200$ points is used, increasing to $N = \lfloor 24 / (\varepsilon / 5) \rfloor + 1$ for larger $\varepsilon$ to ensure sufficient resolution, so that the condition $h \ll \varepsilon$ is satisfied in all cases. The resulting dense matrix eigenvalue problem is solved with \texttt{scipy.linalg.eigh}, restricting to the lowest $\lfloor \varepsilon^{-1} \rfloor + 2$ eigenvalues to target only the bound states, and discarding any eigenvalues above the threshold $\mathcal{E} = 0.999$ to exclude continuum states. Boundary effects are negligible since all mass profiles satisfy $|m(\pm 12)| \approx 1$, placing the domain endpoints well into the classically forbidden region for all eigenvalues of interest.

For each profile we present three diagnostics. The eigenfunction plots compare numerical eigenfunctions against the WKB approximation~\eqref{eq:S4_wkb_final} for representative states, illustrating both the oscillatory structure in the classically allowed region and the exponential decay outside it. The spectrum plots overlay numerical eigenvalues with the WKB predictions from~\eqref{eq:S5_bs} across several values of $\varepsilon$ showing the $\mathcal{O} (\varepsilon)$ agreement and the expected threshold near $\mathcal{E} = 1$. The phase check plots verify the quantization condition directly: for each numerically computed eigenvalue $\mathcal{E}_n$ the phase $\Phi(\mathcal{E}_n) / \pi \varepsilon$ is computed by numerical quadrature and compared against the predicted integer $n + (\sigma_\mathcal{D} + 1)/2$ with residuals confirming agreement well within the $\mathcal{O} (\varepsilon)$ error bound.

\subsection{P\"oschl--Teller: Hyperbolic Tangent Mass Profile}

\begin{figure}[H]
    \centering
    \begin{subfigure}{\linewidth}
        \centering
        \includegraphics[width=\linewidth]{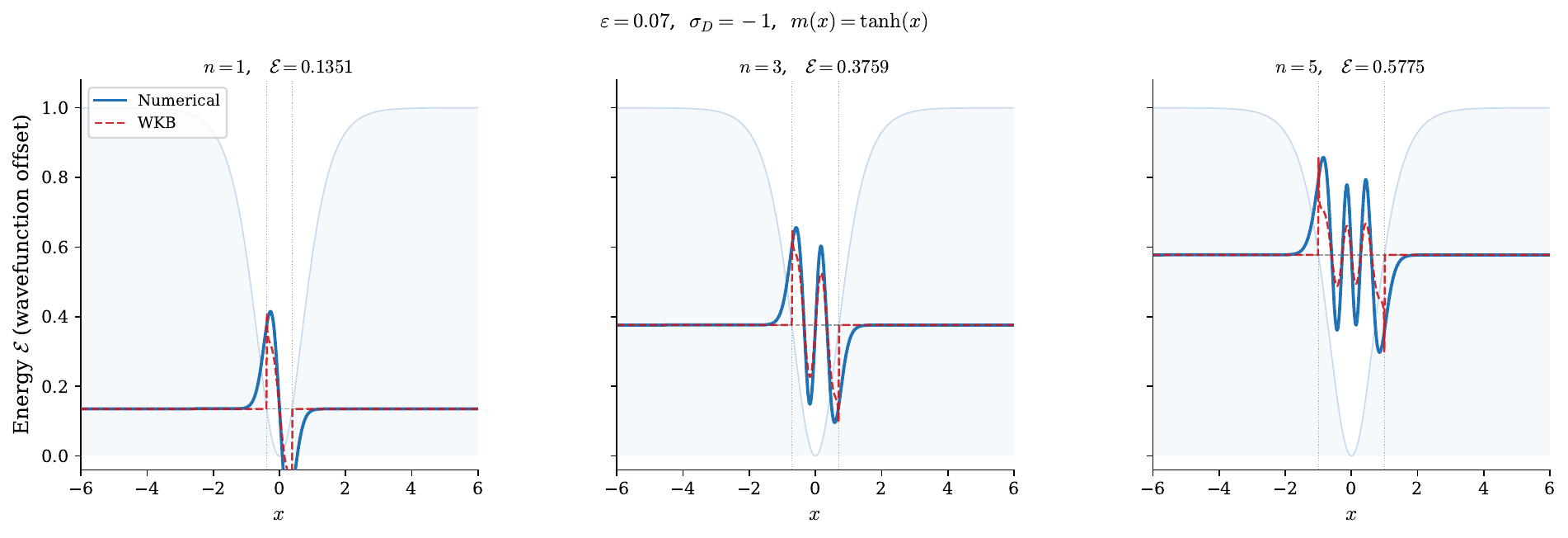}
        \caption{Eigenfunctions (numerical vs.\ WKB) for selected states.}
        \label{fig:tanh_eigenfunctions}
    \end{subfigure}

    \vspace{1em}

    \begin{subfigure}{\linewidth}
        \centering
        \includegraphics[width=\linewidth]{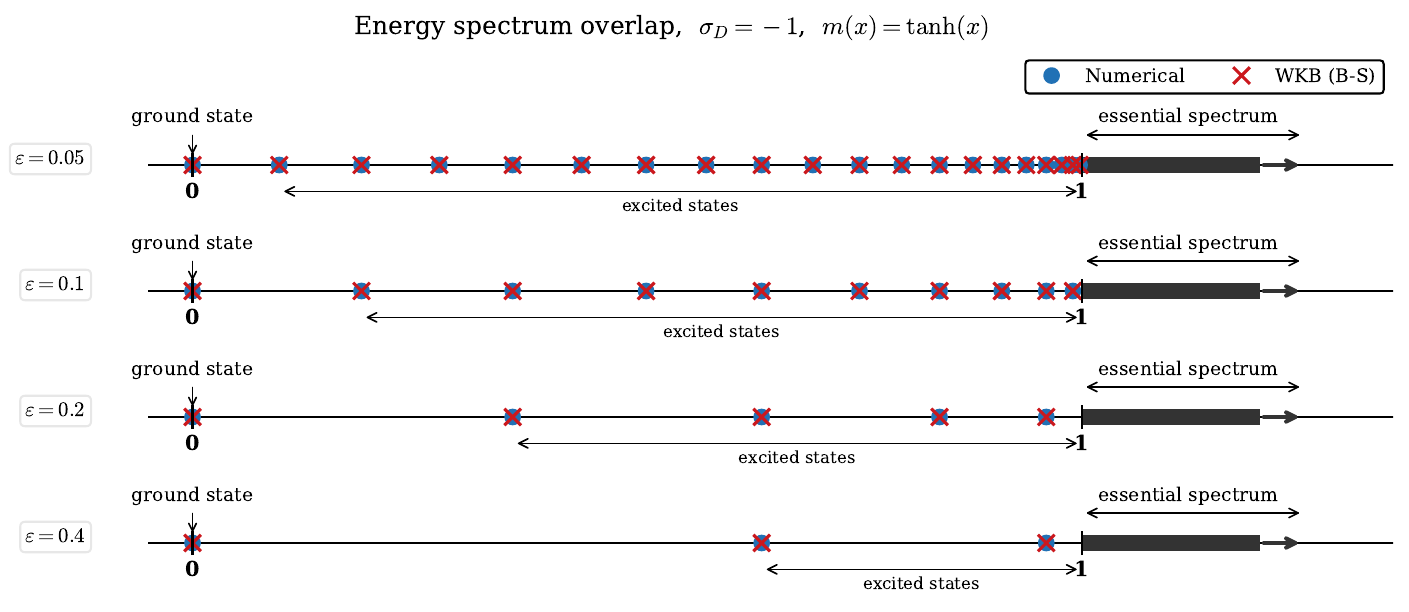}
        \caption{Energy spectrum overlap between numerical and WKB (Bohr--Sommerfeld) eigenvalues across several values of $\varepsilon$.}
        \label{fig:tanh_eigenfunctions}
    \end{subfigure}

    \caption{Results for $m(x) = \tanh(x)$, $\sigma_D = -1$. (a)~Eigenfunctions at representative energy levels; (b)~energy spectrum showing numerical and WKB eigenvalues.}
    \label{fig:tanh_main}
\end{figure}

Figure~\ref{fig:tanh_main} shows exact agreement between numerical and WKB eigenvalues for all $n$, consistent with Proposition~\ref{prop:S6_exact}. The phase check (Figure~\ref{fig:tanh_phase}) confirms this: residuals are of order $10^{-5}$, far below the $\mathcal{O}(\varepsilon)$ bound. 

\begin{figure}[H]
    \centering
    \includegraphics[width=\linewidth]{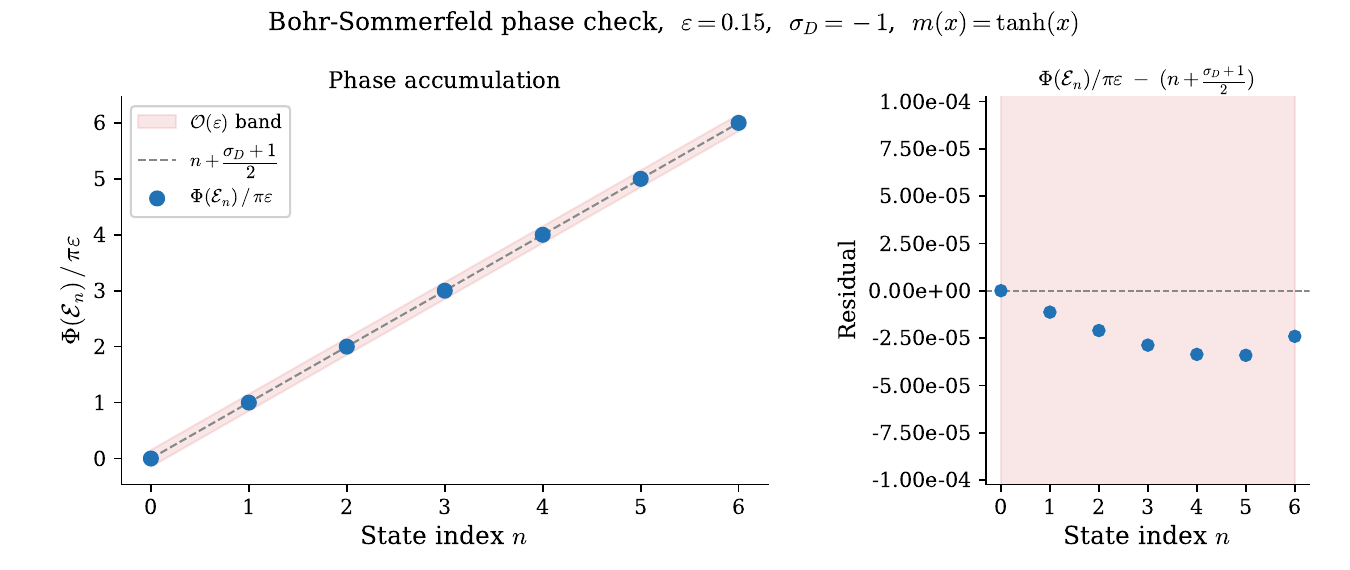}
    \caption{Bohr--Sommerfeld phase check for $m(x) = \tanh(x)$, $\varepsilon = 0.15$, $\sigma_D = -1$. Left: phase accumulation $\Phi (\mathcal{E}_n) / \pi \varepsilon$ vs.\ the predicted $n + (\sigma_D + 1) / 2$. Right: residuals, confirming agreement to $\mathcal{O}(\varepsilon)$.}
    \label{fig:tanh_phase}
\end{figure}

\subsection{Error Function Mass Profile}

\begin{figure}[H]
    \centering
    \begin{subfigure}{\linewidth}
        \centering
        \includegraphics[width=\linewidth]{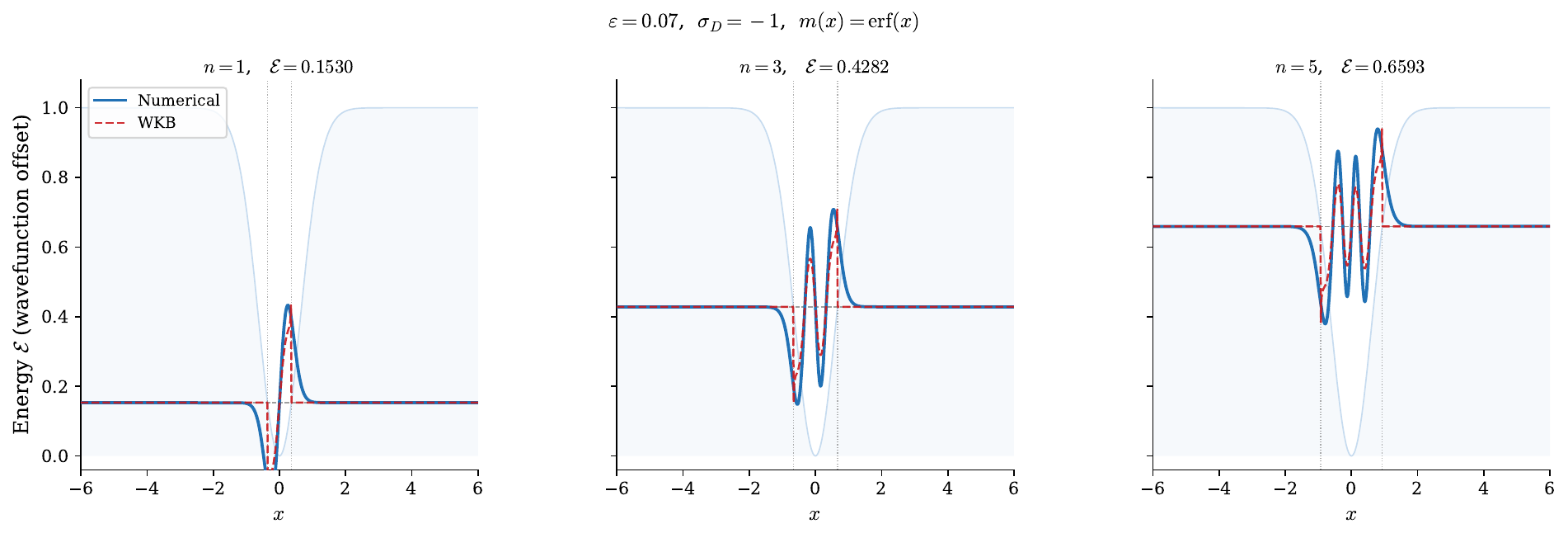}
        \caption{Eigenfunctions (numerical vs.\ WKB) for selected states.}
        \label{fig:erf_eigenfunctions}
    \end{subfigure}

    \vspace{1em}

    \begin{subfigure}{\linewidth}
        \centering
        \includegraphics[width=\linewidth]{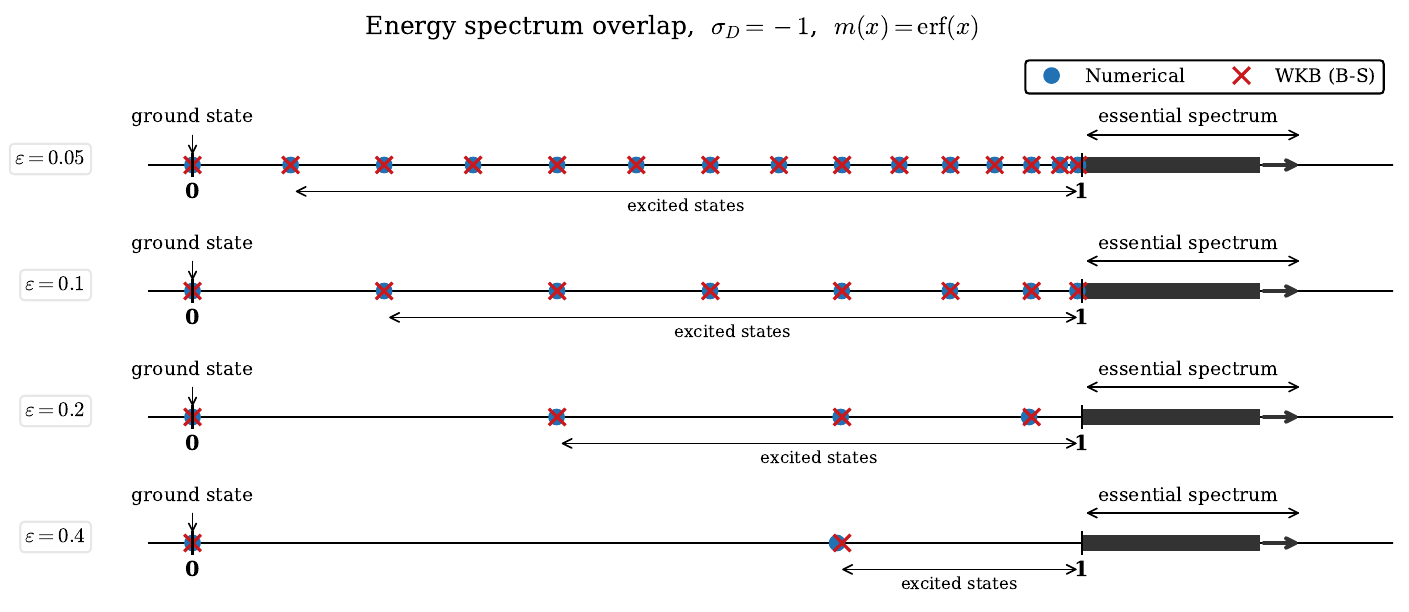}
        \caption{Energy spectrum overlap between numerical and WKB (Bohr--Sommerfeld) eigenvalues across several values of $\varepsilon$.}
        \label{fig:erf_eigenfunctions}
    \end{subfigure}

    \caption{Results for $m(x) = \operatorname{erf}(x)$, $\sigma_D = -1$. (a)~Eigenfunctions at representative energy levels; (b)~energy spectrum showing numerical and WKB eigenvalues.}
    \label{fig:erf_main}
\end{figure}

Figure~\ref{fig:erf_main} shows $\mathcal{O}(\varepsilon)$ agreement between numerical and WKB eigenvalues; here the quantization integral $\Phi(\mathcal{E})$ has no closed  form and is computed numerically via bisection. The phase check residuals (Figure~\ref{fig:erf_phase}) are consistent with the asymptotic error bound.

\begin{figure}[H]
    \centering
    \includegraphics[width=\linewidth]{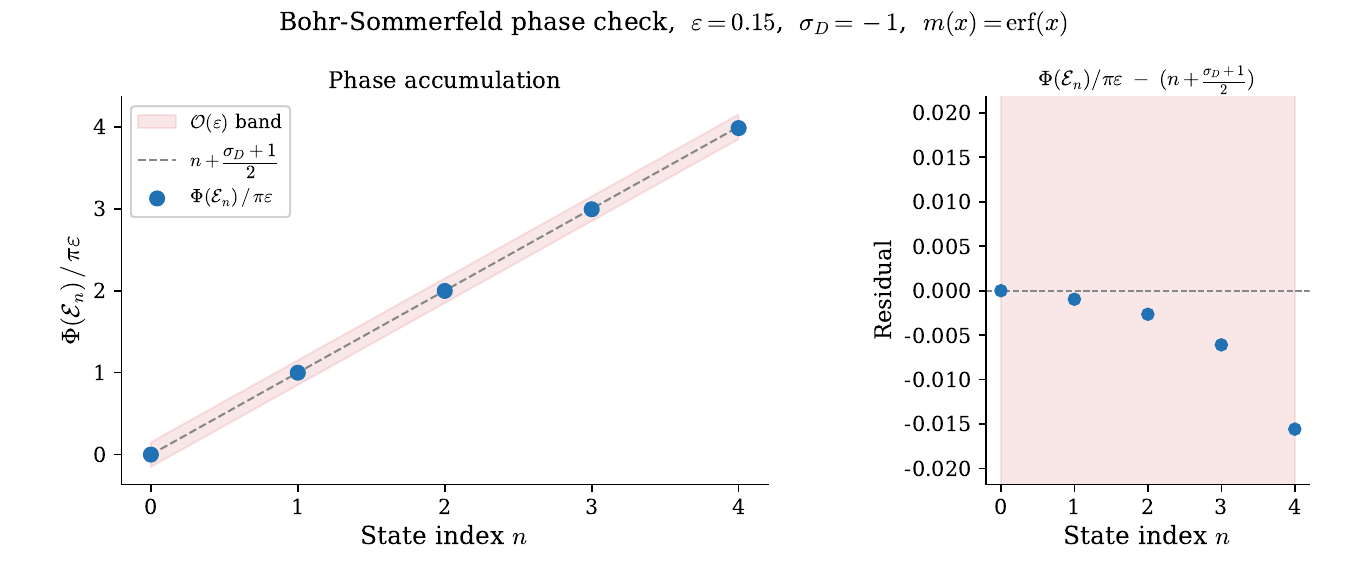}
    \caption{Bohr--Sommerfeld phase check for $m(x) = \operatorname{erf}(x)$, $\varepsilon = 0.15$, $\sigma_D = -1$. Left: phase accumulation $\Phi (\mathcal{E}_n) / \pi \varepsilon$ vs.\ the predicted $n + (\sigma_D + 1) / 2$. Right: residuals, confirming agreement to $\mathcal{O}(\varepsilon)$.}
    \label{fig:erf_phase}
\end{figure}

\subsection{Algebraic Sigmoid Mass Profile}

\begin{figure}[H]
    \centering
    \begin{subfigure}{\linewidth}
        \centering
        \includegraphics[width=\linewidth]{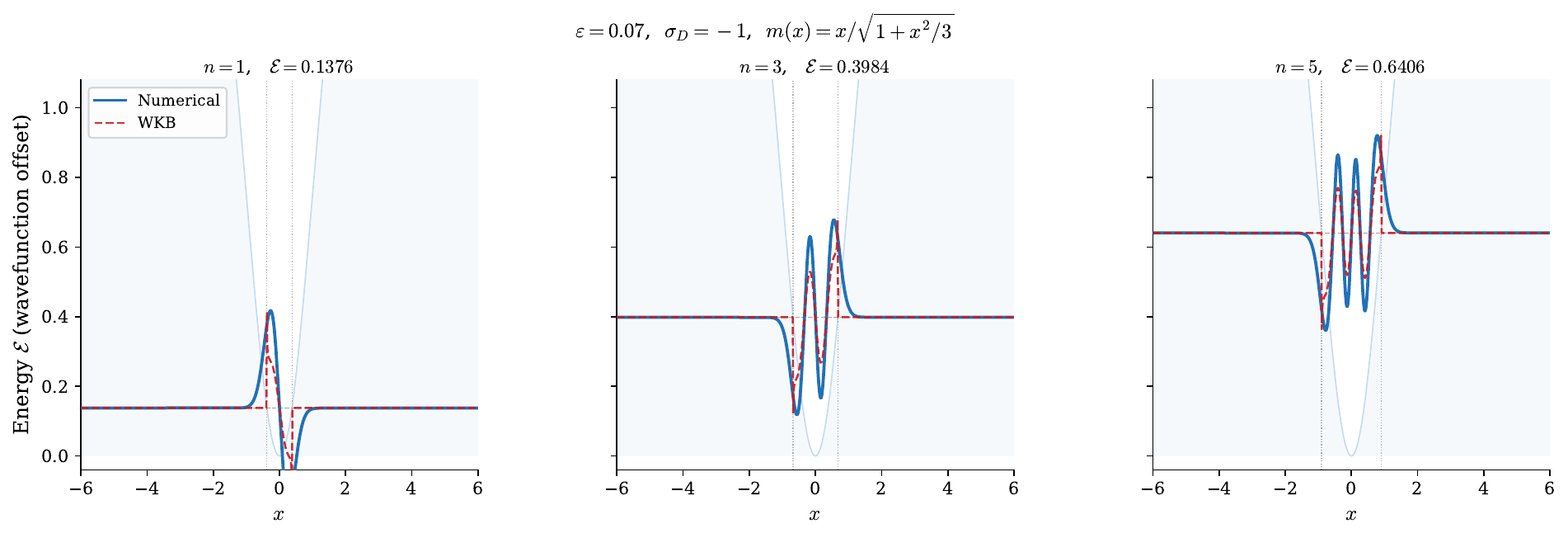}
        \caption{Eigenfunctions (numerical vs.\ WKB) for selected states.}
        \label{fig:algSig_eigenfunctions}
    \end{subfigure}

    \vspace{1em}

    \begin{subfigure}{\linewidth}
        \centering
        \includegraphics[width=\linewidth]{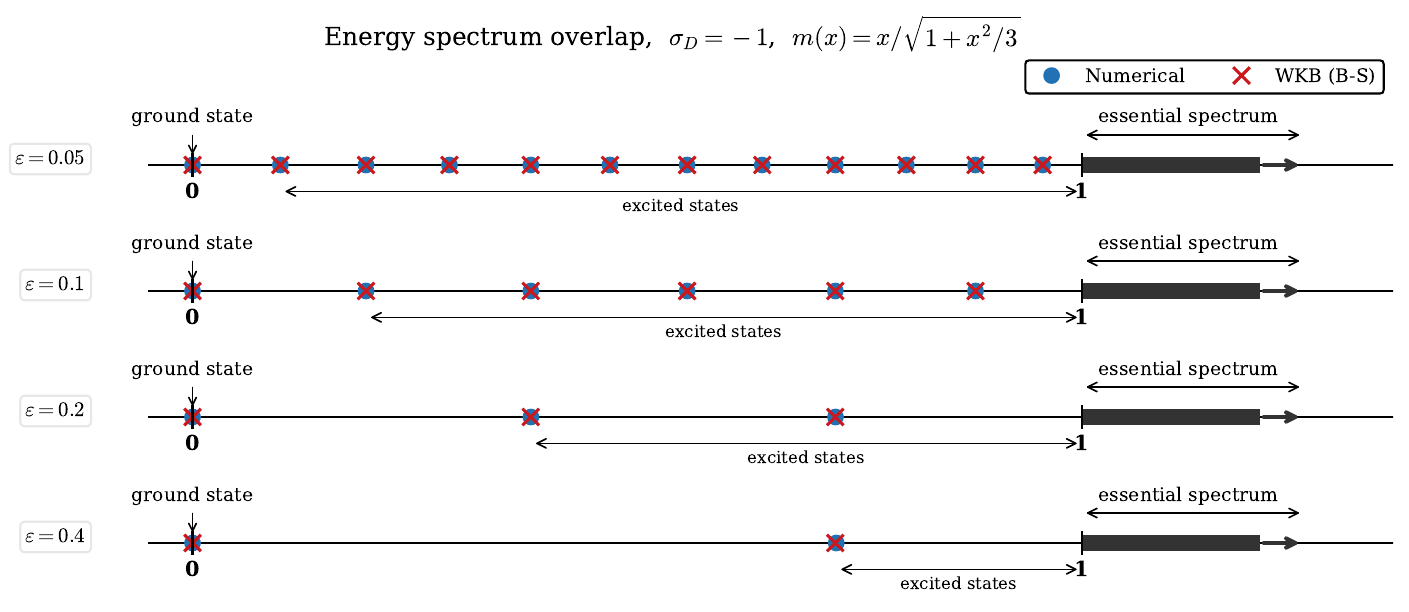}
        \caption{Energy spectrum overlap between numerical and WKB (Bohr--Sommerfeld) eigenvalues across several values of $\varepsilon$.}
        \label{fig:algSig_eigenfunctions}
    \end{subfigure}

    \caption{Results for $m(x) = x / \sqrt{1 + x^2 / 3}$, $\sigma_D = -1$. (a)~Eigenfunctions at representative energy levels; (b)~energy spectrum showing numerical and WKB eigenvalues.}
    \label{fig:algSig_main}
\end{figure}

Figure~\ref{fig:algSig_main} shows the same qualitative behavior as the error function  case. The quantization integral reduces to an elliptic integral and is again evaluated  numerically; the phase check (Figure~\ref{fig:algSig_phase}) confirms $\mathcal{O}(\varepsilon)$  agreement across all computed states.

\begin{figure}[H]
    \centering
    \includegraphics[width=\linewidth]{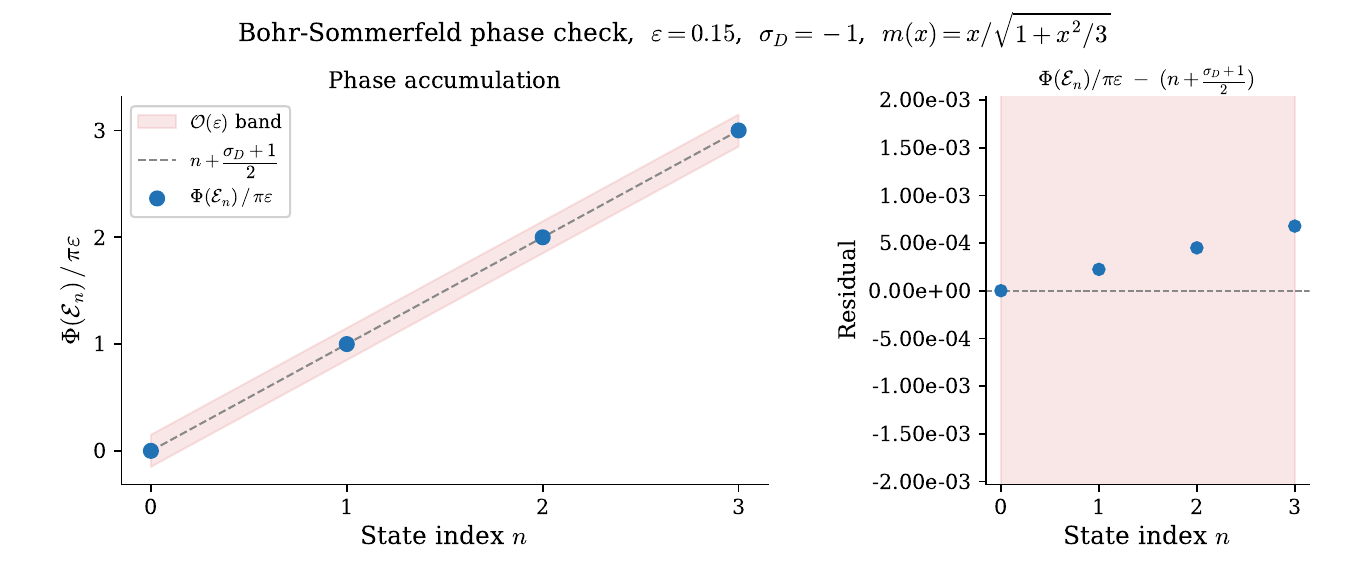}
    \caption{Bohr--Sommerfeld phase check for $m(x) = x / \sqrt{1 + x^2 / 3}$, $\varepsilon = 0.15$, $\sigma_D = -1$. Left: phase accumulation $\Phi (\mathcal{E}_n) / \pi \varepsilon$ vs.\ the predicted $n + (\sigma_D + 1) / 2$. Right: residuals, confirming agreement to $\mathcal{O}(\varepsilon)$.}
    \label{fig:algSig_phase}
\end{figure}

\bibliographystyle{unsrt}
\bibliography{config/references}

\end{document}

%% file: config/preamble.tex
\documentclass[a4paper,11pt]{article}
\usepackage[margin=1in]{geometry}
\usepackage{lmodern}
\usepackage[english]{babel}
\usepackage[utf8]{inputenc}
\usepackage{comment}
\usepackage{amsmath}
\usepackage{amssymb}
\usepackage{graphicx}
\usepackage{subcaption}
\usepackage{float}
\usepackage{pgfgantt}
\usepackage{mathtools}
\usepackage{fancyhdr}
\usepackage{xcolor}
\usepackage{soul}
\usepackage{comment}
\usepackage[colorlinks=true, linkcolor=blue, citecolor=blue]{hyperref}
\usepackage{amsthm}

\theoremstyle{plain}
\newtheorem{theorem}{Theorem}[section]
\newtheorem{lemma}[theorem]{Lemma}
\newtheorem{proposition}[theorem]{Proposition}
\newtheorem{corollary}[theorem]{Corollary}

\theoremstyle{definition}

\newtheorem{remark}[theorem]{Remark}

\usepackage{natbib}

%% file: config/commands.tex
\newcommand{\diff}{\mathop{}\!\mathrm{d}}

\DeclareMathOperator{\sech}{sech}